\documentclass[letterpaper,11pt]{article}

\usepackage[margin=1in]{geometry}
\usepackage{color}
\usepackage{tikz}
\usepackage{authblk}

\usepackage{booktabs} 
\usepackage[utf8]{inputenc}
\usepackage{tikz}
\usepackage{pgfplots}
\usepackage{enumitem}
\usepackage{amssymb}
\usepackage{amsmath}
\usepackage{amsthm}
\usepackage{wrapfig}
\usepackage{algorithm}
\usepackage{algorithmic}
\usepackage{caption}
\usepackage{dsfont}
\usepackage{hyperref}

\newcommand{\bcomment}[1]{}
\newcommand{\TODO}[1]{}

\newcommand{\calO}{\mathcal{O}}

\newcommand{\define}{\triangleq}

\newcommand{\leql}[1]{\overset{(#1)}{\leq}}

\newcommand{\pth}[1]{\left(#1\right)}
\newcommand{\croc}[1]{\left[#1\right]}
\newcommand{\acc}[1]{\left\{#1\right\}}
\newcommand{\norm}[1]{\left|#1\right|}

\newcommand{\condr}[2]{\left[#1\left|#2\right.\right]}
\newcommand{\ceil}[1]{\left\lceil#1\right\rceil}

\newcommand{\Bpth}[1]{\Big(#1\Big)}

\newcommand{\Bacc}[1]{\Big\{#1\Big\}}

\newcommand{\natS}{\mathbb{N}}

\newcommand{\realS}{\mathbb{R}}

\newcommand{\diffS}{\setminus}
\newcommand{\prob}[1]{P\left[#1\right]}
\newcommand{\E}[1]{\mathbb{E}\left[#1\right]}
\newcommand{\ind}[1]{\mathds{1}\{#1\}}

\newcommand{\eh}{\eta}

\newcommand{\mcE}{\mathcal{E}}

\newcommand{\Ga}{G_a}
\newcommand{\Gb}{G_b}
\newcommand{\Ha}{H_a}
\newcommand{\Hb}{H_b}
\newcommand{\Va}{V_a}
\newcommand{\Vb}{V_b}
\newcommand{\rem}{V \setminus H}
\newcommand{\Ea}{E_a}
\newcommand{\Eb}{E_b}

\newcommand{\comp}[1]{\overline{#1}}

\newcommand{\da}[1]{d_a\left(#1\right)}
\newcommand{\db}[1]{d_b\left(#1\right)}
\newcommand{\dxa}[1]{d_{\overline{a}}\left(#1\right)}
\newcommand{\dxb}[1]{d_{\overline{b}}\left(#1\right)}
\newcommand{\dG}[1]{d_G\left(#1\right)}

\newcommand{\daprime}[1]{d_a^{#1}}
\newcommand{\dxaprime}[1]{d_{\overline{a}}^{#1}}
\newcommand{\dbprime}[1]{d_b^{#1}}

\newcommand{\degseq}{\delta}

\newcommand{\Na}[1]{N_a\left(#1\right)}
\newcommand{\Nb}[1]{N_b\left(#1\right)}

\newcommand{\Sa}[1]{\operatorname{sig}_a(#1)}
\newcommand{\Sb}[1]{\operatorname{sig}_b(#1)}

\newcommand{\SG}[1]{\operatorname{sig}_G(#1)}
\newcommand{\Sx}[2]{\operatorname{sig}_{#1}\left(#2\right)}
\newcommand{\SB}[1]{\operatorname{sig}_B'\left(#1\right)}
\newcommand{\SBa}[1]{\operatorname{sig}_a'\left(#1\right)}
\newcommand{\SBb}[1]{\operatorname{sig}_b'\left(#1\right)}

\newcommand{\vvec}{\mathbf{v}}
\newcommand{\wvec}{\mathbf{w}}
\newcommand{\deltavec}{\mathbf{\delta}}

\newcommand{\pvec}{\mathbf{p}}

\newcommand{\pll}{p_{11}}
\newcommand{\plo}{p_{10}}
\newcommand{\pol}{p_{01}}
\newcommand{\poo}{p_{00}}

\newcommand{\plx}{p_{1*}}
\newcommand{\pox}{p_{0*}}
\newcommand{\pxl}{p_{*1}}
\newcommand{\pxo}{p_{*0}}

\newcommand{\plox}{\frac{p_{10}}{p_{1*}}}
\newcommand{\polx}{\frac{p_{01}}{p_{0*}}}

\newcommand{\qo}{q_0}
\newcommand{\ql}{q_1}

\newcommand{\blue}[1]{\textcolor{blue}{#1}}

\newcommand{\highf}{f}

\newtheorem{theorem}{Theorem}
\newtheorem*{theorem*}{Theorem}
\newtheorem{lemma}[theorem]{Lemma}
\newtheorem{corollary}[theorem]{Corollary}
\newtheorem{definition}{Definition}
\newtheorem{remark}{Remark}

\DeclareMathOperator*{\argmin}{argmin}

\begin{document}
\title{Analysis of a Canonical Labeling Algorithm for the Alignment of Correlated Erdős-Rényi Graphs}
\author[1]{Osman Emre Dai\thanks{oedai@gatech.edu}}
\author[2]{Daniel Cullina\thanks{dcullina@princeton.edu}}
\author[1]{Negar Kiyavash\thanks{negar.kiyavash@ece.gatech.edu}}
\author[3]{Matthias Grossglauser\thanks{matthias.grossglauser@epfl.ch}}
\affil[1]{Georgia Institute of Technology, Department of Industrial \& Systems Engineering}
\affil[2]{Princeton University, Department of Electrical Engineering}
\affil[3]{Georgia Institute of Technology, Department of Electrical and Computer Engineering and Department of Industrial \& Systems Engineering}
\affil[4]{École Polytechnique Fédérale de Lausanne, School of Computer \& Communication Sciences}
\date{}
\maketitle

\begin{abstract}
Graph alignment in two correlated random graphs refers to the task of identifying
the correspondence between vertex sets of the graphs. Recent results have characterized the exact information-theoretic threshold for graph alignment in correlated Erdős-Rényi graphs. However, very little is known about the existence of efficient algorithms to achieve graph alignment without seeds. 

In this work we identify a region in which a straightforward $\calO(n^{11/5} \log n )$-time canonical labeling algorithm, initially introduced in the context of graph isomorphism, succeeds in aligning correlated Erdős-Rényi graphs. The algorithm has two steps. In the first step, all vertices are labeled by their degrees and a trivial minimum distance alignment (i.e., sorting vertices according to their degrees) matches a fixed number of highest degree vertices in the two graphs. Having identified this subset of vertices, the remaining vertices are matched using a alignment algorithm for bipartite graphs.
\end{abstract}

\section{Introduction}

Graph alignment (GA) (also called network reconciliation) refers to a class of computational techniques to identify node correspondences across related networks based on structural information.
GA has applications in a variety of domains, including data fusion \cite{tian:tale,zhang:sapper}, privacy \cite{erkip, community1, community2} and in computational biology \cite{Singh:2008,kuchaiev2010topological,malod2015graal,saraph2014,Aladag:2013}. 
For example, in computational biology, a coarse description of the metabolic machinery of a particular species is via a protein-protein interaction (PPI) network, which essentially captures which protein can react with which other protein in that species.
Across species, the PPI networks tend to be strongly correlated, because evolution transfers metabolic processes from species to species.
Therefore, by identifying correspondences among proteins in different species (so-called {\em orthologs}), one is able to transfer biological knowledge from one species to the other.
However, crucially, the actual proteins tend to be chemically different across species, because random mutations alter these proteins over time without affecting their function.
It is therefore not possible to find correspondences between proteins in different species simply by examining their amino-acid sequences.
GA computes such correspondences by exploiting the correlation across networks in different species.

A similar challenge arises in social networks: suppose a set of users have accounts in several social networks.
It is plausible that their links in these networks would be correlated, in the sense that given $u$ and $v$ are linked
in the first network, it makes it conditionally more likely that they are connected in the second.
This can help network reconciliation (e.g., if one wants to create a single network out of several component networks),
and it can hurt privacy (e.g., by exploiting one public network to de-anonymize a private network whose node identities have been
obfuscated).

While a lot of prior work on GA is heuristic in nature, a clean mathematical treatment of the problem first posits a stochastic model over two random graphs.
One parametrization of this model assumes a generator graph $G$, and then generates two correlated observable graph $G_{a,b}$ by sampling the edge set of $G$ twice, independently.
An equivalent formulation, adopted in this paper, considers a joint distribution that generates both graphs without the assumption of an underlying true graph.
Given this random graph model, we can recover the perfect alignment as the matching of the vertex sets under the assumption that pairs of vertices in one graph tend to be adjacent if and only if their true matches are adjacent in the other graph. This can be considered as a generalization of the problem of identifying graph isomorphisms, which corresponds to matching graphs where edges are not just likely but certain to be the same in both graphs.


In this paper, to the best of our knowledge, we present the first algorithm that possesses the following advantages: (i) it is seedless, i.e., it does not require side-information in the form of pre-aligned pairs to operate; (ii) under a well-studied stochastic graph model, the regime where the algorithm achieves perfect alignment can be characterized; and (iii) the algorithm incurs an $O(n^{11/5} \log n )$ computational cost  in the size of the graph, enabling the alignment of large networks.

The algorithm proceeds in two phases: during the first phase, for a fixed threshold parameter $h$, the $h$ highest-degree vertices in both graphs are matched in the natural way (highest degree to highest, second-highest to second-highest, and so forth). For convenience, we call these `anchors'.
In the second phase, each remaining vertex is labeled with a binary vector of length $h$ that encodes its adjacency to the set of anchor vertices. 
The final alignment is then generated via a minimum-distance matching over the labels in both graphs.
Note that the second phase is equivalent to the matching of two bipartite graphs given the matching of one of their partite sets.

We evaluate the performance of the algorithm on the correlated random graph model of asymptotic size and determine conditions for the reliable performance of the algorithm.
This result relies on an achievability result on the matching of bipartite graphs as an intermediary step, which is of independent interest.

The remainder of the paper is organized as follows: In \hyperref[relatedWork]{Section \ref*{relatedWork}}, we survey the relevant prior work on the problem of graph matching in large networks. In \hyperref[model]{Section \ref*{model}}, we introduce our notation, formalize the problem, and present our model of correlated graphs and correlated bigraphs. In \hyperref[results]{Section \ref*{results}}, we state our main result, present the conditions on the successful performance of the two steps of the algorithm, and finally provide the proof for our main result. Then in \hyperref[simulations]{Section \ref*{simulations}}, we compare the algorithm with other known algorithms from the literature. In \hyperref[conclusion]{Section \ref*{conclusion}}, we suggest some directions for future work. We present the proofs of all of our intermediary results in the appendix.

\section{Related Work}
\label{relatedWork}

The graph alignment problem has been studied in a diverse set of fields and with different applications in mind.
First, a line of work focuses on GM as a mode of attack on private information.
An adversary tries to de-anonymize a network that is publicly released, but where node identities have been deliberately obfuscated.
Obviously, there are also legitimate applications for GM: for example, similar approaches have been proposed to reconcile databases, by aligning their database schema \cite{tian:tale,zhang:sapper}.
One such scenario considers the possibility of manipulating the network prior to its release, such that an identifiable sub-network is created \cite{whereforeArtThou} through a form of ``graph steganography''.
In another scenario, the attacker uses queries to attempt to locate the node of a given user \cite{erkip}.
Yet other scenarios assume the availability of some kind of side information, such as community assignments \cite{community1,community2} or subsets of identified vertices (seeds) \cite{seeds,netflix,witness, mossel2019seeded}.
One important method making use of such side information is the so-called percolation method, which starts from the seeds vertices to iteratively grow the alignment until the whole graph is identified \cite{percolation1}, \cite{percolation2}.

In computational biology, PPI network alignment algorithms typically rely on both structural and biological information (in particular, the amino acid sequences of the proteins).
Many heuristics have been developed, which typically try to minimize a cost function that is a convex combination of structural similarity and of sequence similarity.
A few prominent examples include IsoRank \cite{Singh:2008}, the GRAAL family
\cite{kuchaiev2010topological,malod2015graal}, MAGNA and its successor MAGNA++
\cite{saraph2014}, and SPINAL
\cite{Aladag:2013}.
All of these methods are purely heuristic in nature, and have been evaluated without the availability of a ground truth.
Their relative merits are a matter of ongoing debate in the computational biology community.

We show in this paper that efficient graph alignment is possible without any side information.
Henderson et al. propose one such method that performs alignment based on expressions of structural features of vectors \cite{features}.
The proposed features are of two kinds: neighborhood features, constructed only using information on immediate neighbors of the vertex, and recursive features, which include information from a wider region of around the vertex with every iteration.
Also, \cite{pedarsani2013bayesian} presents a heuristic that builds a alignment in phases; matched nodes in one phase serve as distance fingerprints for additional nodes in the next phase.

Non-iterative approaches for graph alignment have also been suggested. Recently a quasi-polynomial time algorithm has been proposed by Boaz et al. that performs alignment by locating copies of some low-likelihood subgraphs in both graphs and using these as the basis of the alignment \cite{barak2018nearly}. We especially note the study by Mitzenmacher and Morgan \cite{mitzenmacher} that proposes performing graph alignment based on algorithms to determine graph isomorphisms.
Defining the problem of graph alignment as a generalization of the isomorphism problem, it becomes possible to attempt to align graphs using some very efficient algorithms developed for the setting of isomorphic graphs.
We consider one such algorithm. Mitzenmacher and Morgan analyze an adversarial setting in which a small number of edge differences are introduced and the algorithm is required to succeed in all cases.
In contrast, we are interested in the case where edge differences are generated at random and the algorithm succeeds with high probability.

Studies on the information-theoretical bound of the graph alignment problem first given by Pedarsani et al. \cite{pedarsani} and further developed by Cullina et al. \cite{sigmetrics}, \cite{DBLP:journals/corr/abs-1711-06783} have established conditions beyond which no algorithm can succeed.
These fundamental bounds provide the main benchmark against which our algorithm will be compared below.


\section{Model}
\label{model}
\subsection{Notation}
For a graph \(G\) we denote its vertex set and edge set as \(V(G)\) and \(E(G)\), respectively.
Alternatively we write \(G=(V;E)\) where \(V=V(G)\) and \(E=E(G)\).
For a bipartite graph \(H\) we denote \(H = (A,B;E)\) where \(A\) and \(B\) are the partite sets and \(E=E(H)\).
For any vertex \(v\in V(G)\) let \(N_G(v)\) be the set of its neighbors in \(G\), \(d_{G}(v)\) its degree and \(d_{\overline{G}}(v)\) its complementary degree.
The maximum degree in graph \(G\) is denoted by \(\Delta(G)\).
When referring to graphs distinguished by their subscript (e.g. \(\Ga\), \(\Gb\)), we use a shorthand notation to denote neighborhoods, degrees etc.
as follows: \(\Na{v}=N_{\Ga}(v)\), \(\da{v}=d_{\Ga}(v)\), \(\dxa{v}=d_{\comp{\Ga}}(v)\).
For a set $X$, let \(X^k\) be set of vectors of length \(k\) with entries from \(X\).
We will use \([k]\) as the index set for these vectors.
We denote vectors in lower case bold font, e.g. \(\vvec=(v_1,v_2,\cdots,v_k)\in V^k\). 

For any \(n\in\natS\),  \([n]\) denotes the set of all integers from \(1\) to \(n\).
We denote by \(\textrm{Bin}(n;p)\) the binomial distribution with \(n\) trials and event probability \(p\).

\subsection{Problem Definition}

Let \(\Ga=(\Va;\Ea)\) and \(\Gb=(\Vb;\Eb)\) be graphs and let \(M : \Va \to \Vb\) be a bijection between their vertex sets.
We say that these graphs are correlated if the edge set of one provides information about the edge set of the other.
We are interested in the case of simple positive correlation: conditioning on the event \(\{u,v\} \in \Ea\) makes the event \(\{M(u),M(v)\} \in \Eb\) more likely.
The details of our random graph model are given in \hyperref[subsection:CER]{Section \ref*{subsection:CER}}.

\paragraph{Graph Alignment Problem:} For a pair of correlated random
\(\Ga = (\Va;\Ea)\) and \(\Gb = (\Vb;\Eb)\), recover \(M :\Va \to \Vb\), the bijection between the vertex pairs in the two graphs based on the correlation of the edge sets.


\subsection{Alignment by Canonical Labeling}


The classical graph isomorphism recovery problem, that is, finding the bijection between vertex sets of a pair of identical graphs, is often solved by canonical labeling based approaches. For a graph \(G = (V;E)\) this approach returns a function \(\ell_G\) from a set of vertices \(V\) to a set of labels called the canonical labeling of vertices, with the property that, any for any permutation \(\sigma\) of the vertex set and the graph \(H\) induced by this permutation on \(G\), \(\ell_G(v) = \ell_H(\sigma(v))\) for all vertices \(v\in V\). In other words, the canonical labeling only depends on the structure of the graph and is invariant to permutations of the vertex set. This allows us to identify an underlying bijection. If \(\ell_G\) is injective, then the labeling allows for recovery of the automorphism.

If the canonical labeling scheme is robust in the sense that small differences in the structure of the graph induce small perturbations on the labels of vertices, then the canonical labeling can still be used to align a pair of graphs that are “adequately” correlated. In this setting, we seek to find a matching
between the label sets of the two graphs that minimizes an appropriately defined labeling distance.

Labeling is done in two steps: In the first step vertices are labeled with their degrees and the small subset of the vertices with high-degrees are identified. These are referred to as `anchors' and form a basis for the alignment of the rest of the graph.
In the second step, the remaining vertices are labeled with signature vectors based on their adjacencies with the anchors identified in step one. 

This second step ignores all edges between unidentified vertices, effectively treating the graph as a bipartite graph.
Therefore, the second step may be considered separately as an algorithm to align two bipartite graphs with one unidentified partite set.
In the remainder of this paper, we refer to the first step as the {\em anchor alignment algorithm} and the second step as the {\em bipartite alignment algorithm}.

\label{subsec-alg}

\begin{algorithm}[H]\captionsetup{labelfont={sc,bf}}

 \caption{
  AnchorSignAlign\\
  {\bf Input:} \(\Ga=(\Va;\Ea)\), \(\Gb=(\Vb;\Eb)\), \(h\) \\
  {\bf Output:} Estimated alignment \(\widehat{M}: \Vb \to \Va\)
 }

\begin{algorithmic}[1]
\label{alg}
\begin{small}

\STATE \textbf{Step 1: Anchor alignment}
\STATE{\(\wvec_{a} = f_{h}(\Ga)\)}
\STATE{\(\wvec_{b} = f_{h}(\Gb)\)}
\FOR {\(i \in [h]\)}
\STATE{\(\widehat{M}(w_{b,i}) = w_{a,i}\)}
\ENDFOR
\STATE  \textbf{Step 2: Bipartite alignment}\\
\STATE{\(\Ha = \acc{w_{a,i} : i \in [h]}\)}
\STATE{\(\Hb = \acc{w_{b,i} : i \in [h]}\)}
\FOR {vertex \(v\in \Vb\setminus H_b\)}
\STATE{\(\widehat{M}(v) = \argmin_{u\in \Va\setminus H_a} \norm{\Sa{u}-\Sb{v}}\)}
\ENDFOR
\end{small}
\end{algorithmic}
\end{algorithm}

The alignment algorithm uses the same canonical labeling scheme originally presented for the graph isomorphism problem by Babai, Erd\H{o}s, and Selkow \cite{isomorphism} and subsequently used for graph alignment in the adversarial setting \cite{mitzenmacher}.
(Note that the graph isomorphism algorithm runs in \(\calO(n^2)\)-time when graphs because the signature matching step can be accomplished by sorting the signatures.
The variation for noisy signatures requires \(\calO(n^2h)\)-time.)

\begin{definition}
\label{def:degseq}
  For a \(n\)-vertex graph \(G\), let \(\deltavec_{G} = (\delta_{G,1},\cdots,\delta_{G,n})\) be the degree sequence of \(G\) in decreasing order. 
\end{definition}
\begin{definition}
  The high-degree sorting function \(\highf_{h}\) takes as input a graph \(G\) on the vertex set \(V\) and lists the \(h\) highest-degree vertices, sorted by degree.
  More precisely, \(\highf_h(G)\) is some vector \(\wvec = (w_1,w_2,\cdots,w_h) \in V^h\) of distinct vertices such that \(\dG{w_i} = \degseq_{G,i}\).
  \bcomment{
  \begin{align*}
    \wvec = (w_1,w_2,\cdots,w_h) = \highf_{h}(G)
  \end{align*}
  implies that \(w_i \neq w_j\) for all \(i\neq j\) and
  \begin{align*}
    \dG{w_1}\geq \dG{w_2} \geq \cdots \geq \dG{w_h} \geq \dG{u}
  \end{align*}
  for all \(u \in V\setminus\{w_1,w_2,\cdots,w_h\}\).}
\end{definition}

The degree sequence of \(G\) is always uniquely defined.
\(\highf_h(G)\) is uniquely defined only if the first \(h\) entries of \(\degseq_{G}\) are strictly decreasing.
If multiple high-degree vertices have the same degree, \(\highf_h(G)\) lists them in some arbitrary order.

Anchor alignment on graphs \(\Ga\) and \(\Gb\) corresponds to the index-by-index alignment of vertices of \(\highf_{h}(\Ga)\) and \(\highf_{h}(\Gb)\).
We refer to the set of \(h\) vertices that appear in \(f_{h}(\Ga)\) as \(\Ha\), the set of \(h\) vertices that appear in \(f_{h}(\Gb)\) as \(\Hb\), and when they are the same we say \(H_a = H_b = H\).
The bipartite alignment algorithm labels each vertex in \(V_a \setminus H_a\) by a binary vector encoding its adjacency with vertices in \(H_a\).
These labels, which we refer to as signatures, are defined as follows:
\begin{definition}
\label{def:sig}
  Given graph \(G\) and anchor vector \(\highf_{h}(G) = \wvec\) \( = (w_{1},w_{2},\cdots,w_{h})\), the signature function \(\textrm{sig}_G\) takes as input vertex \(u\in V(G)\) and returns the signature label of the vertex such that,
  \begin{align*}
    \SG{u} \in \{0,1\}^{h} \quad \textrm{ and } \quad \SG{u}_i = \ind{\{u,w_{i}\} \in E(G)}
  \end{align*}
  where \(\ind{\cdot}\) denotes the indicator function of an event. We use the shorthand notation \(\Sa{u} = \textrm{sig}_{\Ga}(u)\), \(\Sb{u} = \textrm{sig}_{\Gb}(u)\) when referring to graphs \(\Ga\) and \(\Gb\).
\end{definition}

The bipartite alignment algorithm aligns vertices in \(\rem\) such as to minimize the Hamming distance between pairs of signatures of aligned vertices.
In our analysis we consider a naive approach, aligning each vertex in one graph to the vertex with the closest signature in the other graph.
Notice that any graph alignment approach limited to signatures ignores all information pertaining to edges among the unidentified set of vertices.

The steps of the alignment algorithm are summarized in \hyperref[alg]{Algorithm  \ref*{alg}}. We refer to the estimated alignment as \(\widehat{M}\). We say the algorithm is successful when \(\widehat{M} = M\).

\subsection{Correlated Erd\H{o}s-R\'enyi Graphs}
\label{subsection:CER}
We perform our analysis on correlated Erd\H{o}s-R\'enyi (ER) graphs \cite{sigmetrics}. Under the basic ER model of random graphs, \(G\sim ER(n;p)\) is a random graph on \(n\) vertices where any two vertices share an edge with probability \(p\) independent from the rest of the graph.
Under the correlated graph model, \((\Ga,\Gb)\sim ER\pth{n;(\pll,\plo,\pol,\poo)}\) are a pair of graphs on the same set of \(n\) vertices where the occurrences of an edge \(e=\{u,v\}\) between any pair of vertices \(u,v\) is independent and identically distributed with the following probabilities:

\begin{equation}
  \label{corr-bern}
  (\ind{e \in E(G_a)},\ind{e \in E(G_b)}) = 
  \begin{cases}
    (1,1) & \text{w.p.}\ \pll \\
    (1,0) & \text{w.p.}\ \plo \\
    (0,1) & \text{w.p.}\ \pol \\
    (0,0) & \text{w.p.}\ \poo.
  \end{cases}
\end{equation}
The marginal probabilities are then defined as:
\begin{alignat*}{4}
    \plx &= \pll + \plo & \quad\quad\pxl &= \pll + \pol\\
    \pox &= \pol + \poo & \quad\quad\pxo &= \plo + \poo
\end{alignat*}
We denote the vector of probabilities as \(\pvec=(\pll,\plo,\pol,\poo)\).
Note that all probabilities are functions of \(n\).
We limit our interest to sparse graphs and only consider \(\pvec\) such that \(\lim_{n\rightarrow\infty} p_{00} = 1\).

Two other variations of the correlated Erd\H{o}s-R\'enyi model have appeared in the literature.

\textbf{Subsampling model:}
This generates a pair of correlated graphs via subsampling of a parent graph $G_{\text{parent}} \sim ER(n;r)$.
Each edge in $G_{\text{parent}}$ is then included in $\Ga$ with probability $s_a$ and in $\Gb$ with probability $s_b$.
Each of these $2|E(G_{\text{parent}})|$ edge subsampling events are independent.
This results in \((\Ga,\Gb)\sim ER\pth{n;(\pll,\plo,\pol,\poo)}\) with
\begin{align*}
  \pll &= rs_as_b\\
  \plo &= rs_a(1-s_b)\\
  \pol &= r(1-s_a)s_b\\
  \poo &= 1 - r(s_a + s_b - s_as_b).
\end{align*}
This model appeared in Pedarsani and Grossglauser \cite{pedarsani} in the symmetric case $s_a = s_b$.
Observe that $\frac{\pll}{\plx\pxl} = \frac{1}{r} \geq 1$, so this model can only represent non-negatively correlated graphs.

\textbf{Perturbation model:}
This starts by generating a base graph $G_{\text{parent}} \sim ER(n;r)$.
In the adversarial perturbation model considered by Mitzenmacher and Morgan \cite{mitzenmacher}, $G_a$ and $G_b$ are each created by making up to $d/2$ changes to the edge set of $G_{\text{base}}$.
In the natural randomized version, $G_a$ and $G_b$ differ from $G_{\text{base}}$ at each of the $\binom{n}{2}$ vertex pairs independently with probability $\delta = \frac{d}{n(n-1)}$. 
This results in \((\Ga,\Gb)\sim ER\pth{n;(\pll,\plo,\pol,\poo)}\) with
\begin{align*}
  \pll &= r(1-2\delta) + \delta^2\\
  \plo &= \delta - \delta^2\\
  \pol &= \delta - \delta^2\\
  \poo &= (1-r)(1-2\delta) + \delta^2.
\end{align*}

The models that we have just described generates a pair of graphs on the same vertex set \(V\).
To convert these graphs to a pair of correlated graphs on distinct vertex sets, the vertices of \(G_b\) can be relabeled using the bijection $M : V \to V_b$.
This relabeling hides the association between the vertex sets and makes the alignment recovery problem nontrivial.
For the analysis of \hyperref[alg]{Algorithm \ref*{alg}}, it is more convenient to work with pairs of graphs on the same vertex sets rather than work with $M$ explicitly, so we will do this for the remainder of the paper.

In the case of bipartite graphs we use an analogous model.
We denote the distribution as \(ER\pth{h,k;\pvec}\) for pairs of correlated graphs
with left vertex set of size $h$ and right vertex set of size $k$.
For random bipartite graphs \((B_a,B_b)\sim ER\pth{h,k;\pvec}\), a left vertex \(u\), and a right vertex \(v\), the pair of random variables \((\ind{(u,v) \in E(B_a)},\ind{(u,v) \in E(B_b)})\) have the same distribution as \hyperref[corr-bern]{(\ref*{corr-bern})}.

\subsection{Outline and Intuition for the Analysis}

The two steps of \hyperref[alg]{Algorithm  \ref*{alg}} dictate opposing bounds on the value of the parameter \(h\). The bipartite alignment phase requires distinct signatures, which is guaranteed only if the length of the signature vectors (\(h\)) is large enough. However, the performance of the anchor alignment phase degrades as \(h\) grows larger. Our analysis consists of determining upper and lower bounds on \(h\) and identifying the region for \(h\) which satisfies both bounds.

In  \hyperref[subsec-HD]{Subsection \ref*{subsec-HD}} we present a sufficient condition to perfectly align the \(h\) highest-deegree vertices in correlated ER graphs. This gives an upper bound on \(h\). The result is derived by determining the conditions that guarantee, with high probability, that the \(h\) highest-degree vertices have large enough degree separation. It is then unlikely that any two high-degree vertices have their degree order reversed. Applying the Chernoff bound, we show that a degree separation of \(\sigma\sqrt{\log h}\) is sufficient, where \(\sigma^2 \approx n(\plo+\pol)\) is the variance of a vertex degree in \(\Gb\) given its degree in \(\Ga\). Trivially,  independent of the variance, the degree separation must also be at least 1. Thus we get
\begin{align}
\label{minDegSep}
    \textrm{minimum degree separation} \geq \max\acc{1,\sigma\sqrt{\log h}}.
\end{align}
Combining (\ref{minDegSep}) with a known result on the degree separation of high-degree vertices gives \hyperref[theorem:HDMatching]{Theorem \ref*{theorem:HDMatching}}, which states a sufficient condition on \(h\) for high-degree matching. Ignoring logarithmic terms, this condition can be simply written as
\begin{align*}
    \frac{\sqrt{n\pll}}{\max\acc{1,\sqrt{n(\plo+\pol)}}} \geq \omega\pth{h^2}.
\end{align*}
The intuition behind this result is as follows: given that all vertex degrees are distributed within an interval of size roughly \(\sqrt{n\pll}\), we can partition the range of degrees into \(\sqrt{n\pll}/\max\{1,\sqrt{n(\plo+\pol)}\}\) bins of size equal to the minimum degree separation. Two vertices in the same bin violate the degree separation requirement. If the degrees of the \(h\) high-degree vertices were to be distributed uniformly within this range, then by the birthday paradox, we would need the number of bins to be significantly larger than \(h^2\). Clearly high-degree vertices are not uniformly distributed. Nevertheless a rigorous analysis shows that this rough estimate is accurate in the leading term and differs from the actual necessary condition only in the logarithmic terms.

In order to understand the constraints on the bipartite matching phase, in \hyperref[subsec-bipartite]{Subsection \ref*{subsec-bipartite}} we analyze the closely related problem of correlated random bipartite graphs. We try to match one of the partite sets based on the complete knowledge of the matching of the other partite set. The identified set is of size \(h\). As in \hyperref[alg]{Algorithm  \ref*{alg}}, this matching is done through sparse binary signatures. The signatures of the copies of any vertex in the two graphs have around \(h\pll\) common ones. Thus \(h\geq\Omega(1/\pll)\) is a necessary condition for matching. Applying the Chernoff inequality and the union bound over all \(\binom{n-h}{2}\approx n^2/2\) possible mismatches, we derive the result in \hyperref[rem]{Remark \ref*{rem}} which gives the sufficient condition as
\begin{align*}
    h\geq\frac{2\log n+\omega(1)}{\pll}.
\end{align*}

This problem closely relates to the bipartite alignment phase of \hyperref[alg]{Algorithm  \ref*{alg}}; in both cases we assume to have complete knowledge of the alignment of one partite set (i.e. the set of anchor vertices) and try to match the other side by only considering edges that connect these two sets. In the case of the correlated bipartite distribution, the analysis is straightforward since edge random variables are independent. But in the general case the edge random variables between the high-degree set and the remaining vertices are not independent. Fortunately the dependence is weak and it is possible to handle this issue by requiring the anchor set to be robust to the addition or removal of a pair of vertices. (This simply requires an additional degree separation of 2 between anchor vertices.)


\section{Analyses and results}
\label{results}

Our main result is a condition under which Algorithm 1 successful recovers the true graph alignment.
\newcommand{\mainthm}{%
  Let \(\Ga=(V;\Ea)\) and \(\Gb=(V;\Eb)\) such that \((\Ga,\Gb)\sim ER(n,\pvec)\) where \(\pvec\) is a function of \(n\) with \(p_{11} \leq o(1)\),
  \begin{align*}
    \pll \geq \omega\bigg(\frac{\log^{7/5} n}{n^{1/5}}\bigg) \;\; \textrm{and} \;\; \pol+\plo &\leq o\bigg(\frac{\pll ^5}{\log^6 n}\bigg),
  \end{align*}
  Then \hyperref[alg]{Algorithm \ref*{alg}} with parameter \(h\) such that
  \begin{align*}
      \frac{\log n + \omega(1)}{p_{11}} \leq h \leq \mathcal{O}\pth{\frac{\log n}{p_{11}}}
  \end{align*}
  exactly recovers the alignment between the vertex sets of \(\Ga\) and \(\Gb\) with probability \(1-o(1)\).
  }

\begin{theorem}
\label{theorem:main}
\mainthm
\end{theorem}

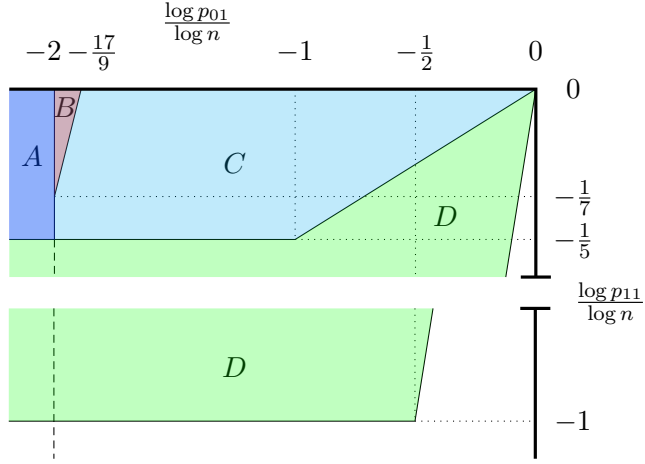
\begin{figure}[ht]
\centering
\begin{tikzpicture}[
 text height = 1.5ex,
 text depth =.1ex,
 b/.style={very thick}
 ]

 \draw (-4.5, 0.9) node {$\frac{\log p_{01}}{\log n}$};
 \draw[color=white] (1,0) node {$\frac{\log p_{11}}{\log n}$};
 
 \draw (0,0.5) node {$0$};
 \draw (-8/5,0.5) node {$-\frac{1}{2}$};
 \draw (-16/5,0.5) node {$-1$};
 \draw (-272/45+0.15,0.5) node {$-\frac{17}{9}$};
 \draw (-32/5-0.15,0.5) node {$-2$};
 
 \draw (0.5,0) node {$0$};
 \draw (0.5,-10/7) node {$-\frac{1}{7}$};
 \draw (0.5,-2) node {$-\frac{1}{5}$};
 
 \draw (-6.7,-0.9) node {$A$};
 \draw (-6.26,-0.25) node {$B$};
 \draw (-4,-1) node {$C$};
 \draw (-1.2,-1.7) node {$D$};

 \draw[dotted] (-16/5,0) -- (-16/5,-2);
 \draw[dotted] (-8/5,0) -- (-8/5,-2.5);
 \draw (-32/5,0) -- (-32/5,-10/7);
 \draw (-32/5,-10/7) -- (-32/5,-2);
 \draw[densely dashed] (-32/5,-2) -- (-32/5,-2.5);
 
 \draw[dotted] (-32/5,-10/7) -- (0,-10/7);
 \draw (-7,-2) -- (-16/5,-2);
 \draw[dotted] (-16/5,-2) -- (0,-2);
 
 \draw (-2*80/50,-2) -- (0,0);
 \draw (-2.5*16/100,-2.5) -- (0,0);
 \draw (-272/45,0) -- (-32/5,-10/7);

 \fill[cyan,opacity=0.25] (-7,0) -- (0, 0) -- (-16/5,-2) -- (-7,-2) -- cycle;
 \fill[blue,opacity=0.25] (-7,0) -- (-32/5,0) -- (-32/5,-2) -- (-7,-2) -- cycle;
 \fill[red,opacity=0.25] (-32/5,0) -- (-272/45,0) -- (-32/5,-10/7) -- cycle;
 \fill[green,opacity=0.25] (-7,-2.5) -- (-2.5*16/100,-2.5) -- (0,0) -- (-16/5,-2) -- (-7,-2) -- cycle;
 \draw[b] (-7,0) -- (0,0) -- (0,-2.5);
 \draw[b] (-0.2,-2.5)--(0.2,-2.5);

\end{tikzpicture}

\begin{tikzpicture}[
 text height = 1.5ex,
 text depth =.1ex,
 b/.style={very thick}
 ]
 \draw (1, 0.1) node {$\frac{\log p_{11}}{\log n}$};
 \draw (0.5,-1.5) node {$-1$};
 \draw (-4,-0.8) node {$D$};
 
 \draw (-7,-1.5) -- (-8/5,-1.5);
 \draw[dotted] (-8/5, 0) -- (-8/5, -1.5);
  
 \draw[dotted] (-8/5,-1.5) -- (0,-1.5);
 \draw[densely dashed] (-32/5,0) -- (-32/5,-2);
 
 \draw (-10*16/100,-1.5) -- (-8.5*16/100,0);
 


 \fill[green,opacity=0.25] (-10*16/100,-1.5) -- (-8.5*16/100,0) -- (-7,0) -- (-7,-1.5) -- cycle;
 \draw[b] (0,0) -- (0,-2);
 \draw[b] (-0.2,0)--(0.2,0);
\end{tikzpicture}

\caption{
  Comparison of regions of achievability for symmetric noise (\(\plo = \Theta\pth{\pol}\)): (A) region achievable by \hyperref[alg]{Algorithm \ref*{alg}} under no expected noise \cite{bollobas}, (A\(\cup\)B) under the adversarial model \cite{mitzenmacher}, (A\(\cup\)B\(\cup\)C) under the random graph model, (A\(\cup\)B\(\cup\)C\(\cup\)D) theoretical achievability region for the random graph model \cite{sigmetrics}.
  }
 \label{regions}
\end{figure}

\hyperref[regions]{Fig. \ref*{regions}} illustrates the asymptotic achievability region of \hyperref[alg]{Algorithm \ref*{alg}} as a function of graph density \(\pth{\frac{\log\pll}{\log n}}\) and noise \(\pth{\frac{\log\plo}{\log n}}\). We also include the achievability of the noiseless scenario \cite{isomorphism}, more challenging adversarial scenario \cite{mitzenmacher}, as well as the information theoretic achievability region \cite{sigmetrics}. We only consider the symmetric case, where \(\plo = \Theta\pth{\pol}\). The \(x\)-axis shows \(\frac{\log \pol}{\log n}\) and the \(y\)-axis shows \(\frac{\log \pll}{\log n}\). Note that in the region \(x < -2\), the number of edge edge differences between the pairs of graphs is zero under the adversarial model and is zero with high probability under the random graph model, so the alignment problem reduces to the graph isomorphism problem.

The adversarial model is defined as follows: Consider a random graph \(\Ga = ER(n;p)\) and its modified copy \(\Gb\) obtained by the addition or deletion of at most \(d\) edges by an adversary where \(d\geq 2\) is a deterministic function of \(n\).

Note that the parameters in this problem formulation relate to the correlated random graph problem through:
\begin{align}
\label{rand2adv}
    p = \plx = \pll + \plo = \pll\pth{1+o(1)} \quad \textrm{ and } \quad d = \pth{\plo + \pol}\binom{n}{2}.
\end{align}
By Theorem 5.3 in \cite{mitzenmacher}, there exists an appropriate choice of parameter \(h\) for which \hyperref[alg]{Algorithm \ref*{alg}} perfectly aligns the vertex sets of the two graphs with probability at least \(1-o(1)\) as long as \(p = \omega\pth{d\log n \pth{\frac{d^2}{ n}}^{1/7}}\). By the equalities in \hyperref[rand2adv]{(\ref*{rand2adv})}, this condition is satisfied when
\[\Omega\pth{n^{-2}} \leq \pth{\pol+\plo}^{\frac{9}{7}} \leq o\pth{\frac{\pll}{\log n}n^{-17/7}}.\]
Recall that the \(x\)-axis shows \(\frac{\log \pol}{\log n}\) and the \(y\)-axis shows \(\frac{\log \pll}{\log n}\).
Taking the logarithm of both sides and dividing by \(\log n\), in the symmetric case, results in the triangular region defined by the inequality
\[-2 \leq \frac{9}{7}x \leq y - \frac{17}{7} - o(1).\]

Note that \(d = o(1)\) for \(\pol+\plo \leq o\pth{n^{-2}}\), so the adversarial scenario with a fixed number of edge changes reduces to the graph isomorphism problem and under the random graph model the graphs are isomorphic with high likelihood. The condition to guarantee successful alignment for that problem, given in Theorem 3.17 in \cite{bollobas}, is \(p = \omega\pth{n^{-1/5}\log n}\), which corresponds to the region where
\begin{align*}
    y \geq -\frac{1}{5} + o(1) \quad \textrm{ and } \quad x \leq -2.
\end{align*}

The achievability region is derived similarly. Theorem 2 in \cite{sigmetrics} gives the following achievability condition as
\[\pll \geq 2\frac{\log n + \omega(1)}{n} \quad \textrm{ and } \pol\plo = o\pth{\pll\poo},\]
which for \(\poo = 1-o(1)\) and \(\plo = \Theta\pth{\pol}\) corresponds to 
\[\pll = \Omega\pth{\frac{\log n}{n}} \quad \textrm{ and } \pol^2 = o\pth{\pll}.\] This gives the the region defined by
\begin{align*}
    y \geq -1 + o(1) \quad \textrm{ and } \quad 2x \leq y.
\end{align*}

In \hyperref[subsec-HD]{subsection \ref*{subsec-HD}} we analyze the performance of the anchor alignment stage of the algorithm.
In \hyperref[subsec-bipartite]{subsection \ref*{subsec-bipartite}} we present the result on the performance of bipartite graph alignment stage of the algorithm.
Finally, in \hyperref[subsec-general]{subsection \ref*{subsec-general}} the results from these two analyses are combined to provide a proof on performance of the alignment algorithm.



\subsection{Anchor alignment}
\label{subsec-HD}
The expected performance on the alignment of the anchors (i.e. high-degree vertices) is a function of the sparsity of the graph, its size, and the number of anchors to be matched.
We first present a result on the required minimum degree separation between a pair of vertices in one graph to guarantee a given degree separation on the other graph with high probability.
We remind the reader of our shorthand notation where for any vertex \(v\in V\), \(\da{v}\) and \(\dxa{v} = |V|-\da{v}-1\) denote \(v\)'s degree and inverse degree in \(\Ga\), respectively. Similarly \(\db{v},\dxb{v}\) denote the degree and inverse degree in \(\Gb\).

\newcommand{\lemmaHDMatchingTwo}{%
    Let \((\Ga,\Gb)\sim ER(n;\pvec)\). Given \(u,v\in V(G)\) such that \(\da{u}>\da{v}\), define \(\varphi \define \da{u}\plox + \dxa{v}\polx\) and \(\varepsilon = \polx + \plox\). Let \(\eh \in(0,\infty)\) be a function of \(n\). If
      \begin{align*}
        &\da{u}-\da{v} \geq (1-\varepsilon)^{-1}\Bpth{k + 4\max\pth{\eh,\sqrt{\varphi\cdot\eh}}},
      \end{align*}
      Then \(\prob{\db{u}-\db{v}\leq k} \leq e^{-\eh}\).
  }

\begin{lemma}
\label{lemma:HDMatching2}
\lemmaHDMatchingTwo
\end{lemma}

Lemma~\ref{lemma:HDMatching2} involves two lower bounds on the gap between degrees: one depending on $\eta$ and the other on $\sqrt{\varphi \cdot \eta}$.
  The quantity \(\varphi\) is the expected number of edges `lost' by \(u\) and `gained' by \(v\) when moving from \(\Ga\) to \(\Gb\).
  A larger \(\varphi\) implies higher likelihood for the degree gap to be `bridged' moving from \(\Ga\) to \(\Gb\).
  At the dense high-noise performance limit, the \(\sqrt{\varphi\cdot\eta}\) lower bound is dominant.
  The \(\eta\) lower bound arises from the discreteness of the degrees.
This bound is dominant at the sparse low-noise limit.

\hyperref[lemma:HDMatching2]{Lemma \ref*{lemma:HDMatching2}} only concerns pairs of vertices. Next we present a condition on the graph sequence of \(\Ga\) that guarantees with high probability the desired degree separation among high-degree vertices in \(\Gb\). Recall that, by \hyperref[def:degseq]{Definition \ref*{def:degseq}}, \(\degseq_a\) and \(\degseq_b\) denote the degree sequences in \(\Ga\) and \(\Gb\) respectively.

\newcommand{\CorollaryHDMatchingN}{%
  Let \((\Ga,\Gb)\sim ER(n;\pvec)\) where \(\Ga = (V;\Ea)\) and \(\Gb = (V;\Eb)\).
  Define \(\varphi \define \Delta(\Ga)\plox + n\polx\) and \(\varepsilon \define \polx + \plox\).
  Let \(h\in[n]\) and
  \(\eh\) be functions of \(n\).
  Let \(s\) be an integer such that \(s \geq h + \frac{1}{\eh}\log \pth{\frac{n}{h}}+1\).
  If
  \begin{align}
  \label{corollary:HDMatchingN:cond1}
    \forall i\in [s], \quad \degseq_{a,i} - \degseq_{a,i+1} &\geq (1-\varepsilon)^{-1}\Bpth{k + 4\max\acc{\eh,\sqrt{\varphi\cdot\eh}}}
  \end{align}
  then, with probability at least \(\frac{1-(2h+1)e^{-\eh}}{1-e^{-\eh}}\), \(f_{h}(\Ga) = f_{h}(\Gb)\) and \(\degseq_{b,i} - \degseq_{b,i+1} > k\) for any \(i\in[h]\).
  }

\begin{corollary}
\label{corollary:HDMatchingN}
\CorollaryHDMatchingN
\end{corollary}

The \(\varphi\) term in \hyperref[corollary:HDMatchingN]{Corollary \ref*{corollary:HDMatchingN}} corresponds to an upper bound for the same term in \hyperref[lemma:HDMatching2]{Lemma \ref*{lemma:HDMatching2}} that we obtain by replacing the vertex degree with the max degree in the graph, and the inverse degree with \(n\). We then need the following upper bound on the maximum degree of a random graph.

\begin{lemma}
  \label{lemma:maximumDegree}
  Let \(G\sim ER(n;p)\) with \(p \geq \omega\pth{\frac{\log n}{n}}\).
  For any constant \(\epsilon > 0\), we have \(P[\Delta(G) \geq pn(1+\epsilon)] \leq o(1)\).
\end{lemma}

\hyperref[corollary:HDMatchingN]{Corollary \ref*{corollary:HDMatchingN}} relies on \(\Ga\) having a degree sequence whose largest terms are sufficiently separated. We now present a condition that guarantees a given degree separation for almost all random graphs.

\begin{theorem}
  \label{theorem:Bollobas}
  (\cite{bollobas} Theorem 3.15) Let \(h\in\natS\) and \(c\in\realS^+\) functions of \(n\) such that \(h=o(n)\) and \(c=o(1)\). Then, with probability \(1-o(1)\), in \(G\sim ER(n,p)\)  
  \[
    \degseq_i - \degseq_{i+1} \geq \frac{c}{h^2}\pth{\frac{np(1-p)}{\log n}}^{1/2} \quad \textrm{for each } i\in [h].
  \]
\end{theorem}

\bcomment{
\begin{lemma}
  \label{lemma:collision}
  Let \(h\in\natS\), \(c \in\natS\), and \(\epsilon \in\realS^+\) functions of \(n\) such that \(h=o(n)\) and \(c=o(1)\). Then, with probability \(1-\epsilon\), in \(G\sim ER(n,p)\)  
  \[
    \degseq_i - \degseq_{i+1} \geq \frac{c \epsilon}{h^2}\pth{\frac{np(1-p)}{\log n}}^{1/2} \quad \textrm{for each } i\in [h].
  \]
\end{lemma}
\begin{proof}
  For vertices \(u\) and \(v\) in \(V(G)\), we need an upper bound on the probability that \(d_G(u) \geq t\), \(d_G(v) \geq t\), and \(|d_G(u) - d_G(v)| < s\).
  We have
  \[
    P[d_G(v) \geq t, |d_G(u) - d_G(v)| < s | d_G(u)] \leq 2s P[d_G(v) = t-1]
  \]
\end{proof}
}

We are now in a position to present a result on the performance of the high-degree matching step of our algorithm.
First we define the three events that are needed to be able to successfully align the high-degree vertices: the set of high-degree vertices must be the same in the two graphs and in each graph the high-degree vertices must have sufficiently separated degrees. Distinct degrees are clearly required, but we require the stronger condition that degrees have difference of at least 3. This allows us to establish the independence of this stage of the algorithm with the bipartite matching stage later in \hyperref[subsec-general]{Subsection \ref*{subsec-general}}.


\begin{definition}
\label{HD-events}

Let \(\mcE^{\operatorname{H}}\) be the event that the lists of the \(h\) highest-degree vertices in \(\Ga\) and \(\Gb\) are the same, i.e. \(f_h(\Ga) = f_h(\Gb)\).
This is the ``high-degree match'' event.
Let \(\mcE_a^{\operatorname{S}}\) be the event that \(\degseq_{a,i} > \degseq_{a,i+1} + 2\) for all \(i \in [h]\).
Define \(\mcE_b^{\operatorname{S}}\) analogously for \(\delta_b\).
These are the ``degree separation'' events. 
\end{definition}
\begin{theorem}
  \label{theorem:HDMatching}
  Let \((\Ga,\Gb)\sim ER(n,\pvec)\) where \(\pvec\) is a function of \(n\) such that \(p_{00}=1-o(1)\). Moreover let \(h\in[n]\) such that \(\omega(\log n) \leq h \leq o(n)\).
  If
  \begin{align}
    \label{limit:h2sqrt-1} \max\acc{(\log h)^2, n(\pol + \plo)\log h} &\leq o\pth{\frac{n\,\pll}{h^4\log n}\cdot\frac{\pll}{\plx}},
  \end{align}
  then \(P\croc{\mcE^H \wedge \mcE_a^{\operatorname{S}} \wedge \mcE_b^{\operatorname{S}}} \geq 1-o(1)\).
\end{theorem}

\begin{proof}
  To apply \hyperref[corollary:HDMatchingN]{Corollary \ref*{corollary:HDMatchingN}}, \(\eta\) and \(s\) must satisfy \(s \geq h + \frac{1}{\eh}\log \pth{\frac{n}{h}}\).
  We pick \(\eta\) such that \(s = \ceil{h + \frac{1}{\eh}\log \pth{\frac{n}{h}}}\) and \(\log h + \omega(1) \leq \eta \leq \mathcal{O}(\log h)\).
  The condition \(h \geq \omega(\log n)\) guarantees that \(s \leq h(1+o(1))\).

  Applying \hyperref[lemma:maximumDegree]{Lemma \ref*{lemma:maximumDegree}} , we have
  \begin{align}
    \label{ineq:rho}
      \varphi = \Delta(\Ga)\plox + n\polx \leq (1+\epsilon)n\plo + n \frac{\pol}{\pox} \leq (1+\epsilon+o(1))n(\plo+\pol).
  \end{align}
  
  Define \(c \define \pth{\frac{s^4\log n}{n\,\plx}}^{1/2}\pth{1-\polx-\plox}^{-1}\pth{2+4\max\{\eta, (r\eta)^{1/2}\}}\). By \(\poo = 1-o(1)\) we have \(\pth{1-\plox-\polx}^{-1} = \pth{\frac{\pll}{\plx}-o(1)}^{-1}=\Theta\pth{\frac{\plx}{\pll}}\). Together with the upper bounds on \(\eta\), \(h\), and \(s\), we get
  \begin{align*}
    c \leq& \mathcal{O}(1)\pth{\frac{ h^4 \log n}{n\pll}\cdot\frac{\plx}{\pll}}^{1/2}\max\acc{\log h, (n(\pol + \plo)\log h)^{1/2}}.
  \end{align*}
  From \hyperref[limit:h2sqrt-1]{(\ref*{limit:h2sqrt-1})}, we have \(c \leq o(1)\).

  By \hyperref[theorem:Bollobas]{Theorem \ref*{theorem:Bollobas}}, with probability \(1-o(1)\), we have a minimum separation of \(2+4\max\{\eta, (\varphi\cdot\eta)^{1/2}\}\) among the top \(s\) degrees in \(\Ga\sim ER(n;\plx)\).
  Then \hyperref[corollary:HDMatchingN]{Corollary \ref*{corollary:HDMatchingN}} implies that the probability that \(f_h(\Ga) \neq f_h(\Gb)\) is at most \(se^{-\eh} \leq (1+o(1))he^{-\log h -\omega(1)} \leq o(1)\).  
\end{proof}


\subsection{Bipartite graph alignment}
\label{subsec-bipartite}

We will need the following method of specifying an induced bipartite subgraph.
Let \(G\) be a graph on the vertex set \(V\) and let \(U \subseteq V\).
Let \(\wvec\) be a vector of \(h\) distinct vertices in \(V\setminus U\).
Define \(G[U,\wvec]\) to be the bipartite graph with left vertex set \(U\), right vertex set \([h]\), and edge set
\[
  E(G[U,\wvec]) = \{(u,j) \in U \times [h] : (u,w_j) \in E(G)\}).
\]

Recall that in \hyperref[alg]{Algorithm \ref*{alg}}, we have \(\wvec_{a} = f_{h}(\Ga)\) and \(\wvec_{b} = f_{h}(\Gb)\). By \hyperref[def:sig]{Definition \ref*{def:sig}}, the signature of any \(u\in U\) is the edge indicator function for \(G_a[\{u\},\wvec_a]\):
\begin{align*}
  \Sa{u} \in \{0,1\}^h \,\, \textrm{ and } \,\, \Sa{u}_i &= \ind{(u,i) \in E(G_a[\{u\},\wvec_a])}.
\end{align*}
We define an analogous signature scheme for bipartite graphs to be used for the bipartite alignment step.

\begin{definition}
  Given the bipartite graph \(B=(V,[h];E)\), the bipartite signature function \(\textrm{sig}_B'\) takes as input vertex \(u\in V\) and returns the signature label of the vertex such that
  \begin{align*}
    \SB{u} \in \{0,1\}^h \quad \textrm{ and } \quad \SB{u}_i = \ind{(u,i) \in E}  
  \end{align*}
  When referring to signatures on bipartite graphs that are distinguished only by their subscripts (e.g. \(B_a\) and \(B_b\)) we only denote the signatures in shorthand notation, e.g. \(\textrm{sig}_a'(u) = \textrm{sig}_{B_a}'(u)\), \(\textrm{sig}_b'(u) = \textrm{sig}_{B_b}'(u)\).
\end{definition}

We restate the second half of \hyperref[alg]{Algorithm \ref*{alg}} as the bipartite graph alignment algorithm in \hyperref[alg:bip]{Algorithm \ref*{alg:bip}}

\begin{algorithm}[H]\captionsetup{labelfont={sc,bf}}
\begin{algorithmic}[1]
\begin{small}
\FOR {vertex \(v \in \Vb\)}
\STATE{\(\widehat{M}(v) = \argmin_{u \in \Va} \norm{\SBa{u}-\SBb{v}}\)}
\ENDFOR
 \caption{
  Bipartite Graph alignment \\
  {\bf Input:} {\(B_a=(V_a,[h];\Ea)\), \(B_b=(V_b,[h];\Eb)\)}\\
  {\bf Output:} Estimated alignment \(\widehat{M}: V_b \to V_a\)
 }
 \label{alg:bip}
 \end{small}
\end{algorithmic}
\end{algorithm}

Suppose that we have bipartite graphs \(B_a=(V_a,[h];\Ea)\) and \(B_b=(V_b,[h];\Eb)\) such that \(|V_a|=|V_b|\). Assume there is an exact correspondence between the vertex sets, expressed by the alignment \(M:V_b \to V_a\). \hyperref[alg:bip]{Algorithm \ref*{alg:bip}} is guaranteed to map vertex \(u\in V_b\) to \(M(u)\in V_a\) if
\begin{align}
\label{cond:bipNecessary}
    \norm{\SBa{M(v)} - \SBb{u}} > \norm{\SBa{M(u)}-\SBb{u}}
\end{align}
for any \(v\in V_b\diffS \{u\}\).
Hence verifying the equality above for any ordered pair of vertices \((u,v)\in V_b^2\) guarantees that the algorithm perfectly aligns all vertices.

In the remainder of the section, in order to avoid cumbersome notation, we assume that, without loss of generality, \(V_a = V_b = V\) and the true alignment is the trivial alignment \(M(v) = v\) for any \(v\in V\).

To analyze \hyperref[alg:bip]{Algorithm \ref*{alg:bip}} for random bipartite graphs, we need the following lemma which bounds the probability that a pair of vertices are misaligned. This corresponds to the failure of \hyperref[cond:bipNecessary]{(\ref*{cond:bipNecessary})} for either one of the vertices.
\begin{lemma}
  \label{lemma:bip-error}
  
  Let bipartite graphs \(B_a = (\{u,v\},[h];\Ea)\) and \(B_a = (\{u,v\},[h];\Eb)\) be distributed according to \((B_a,B_b)\sim ER(2,h;\pvec)\).
  
  Define \(\mcE^{\operatorname{M}}(B_a,B_b)\) to be the ``misalignment event'' i.e. the event where either of the following inequalities hold:
  \begin{align*}
      &\norm{\SBa{v} - \SBb{u}} \leq \norm{\SBa{u}-\SBb{u}}\\
      \textrm{ or } \quad &\norm{\SBa{u} - \SBb{v}} \leq \norm{\SBa{v}-\SBb{v}}.
  \end{align*}
  Then \(\prob{\mcE^{\operatorname{M}}(B_a,B_b)} \leq 2\exp\pth{-h\rho^2}\) where 
  \begin{align*}
      \rho \define \sqrt{\poo\plx + \pll \pox}-\sqrt{\plo\pox + \pol\plx}.
  \end{align*}
\end{lemma}

The quantity \(\rho\) is a measure of the correlation between the pair of graphs. The likelihood of misalignment between a pair of vertices can be upper bounded in terms of \(h\), the size of the readily identified set, and \(\rho\), the strength of the correlation between the new graphs. Applying this result over the entire graph gives us the following result.

\begin{remark}
  \label{rem}
  Let \((B_a,B_b)\sim ER(n,h;\pvec)\).
  Then for each \(u,v \in [n]\), the subgraphs induced by \(\{u,v\}\) and \([h]\) have joint distribution \(ER(2,h;\pvec)\).
  By \hyperref[lemma:bip-error]{Lemma \ref*{lemma:bip-error}}, the probability that \hyperref[alg:bip]{Algorithm \ref*{alg:bip}} misaligns \(u\) with \(v\) or \(v\) with \(u\)  is at most \(2\exp\pth{-h\rho^2}\).
  Then, by the union bound over all \(\binom{n}{2}\) pairs of vertices, \hyperref[alg:bip]{Algorithm \ref*{alg:bip}} correctly recovers the alignment between \(B_a\) and \(B_b\) with probability at least \(1-n(n-1)\exp\pth{-h\rho^2}\) and the algorithm is correct with probability $1-o(1)$ when 
  \begin{align}
      h \geq \frac{2 \log n + \omega(1)}{\rho^2}.
  \end{align}
\end{remark}

In our analysis of \hyperref[alg]{Algorithm \ref*{alg}}, the situation is similar yet not quite as simple as the one described in \hyperref[rem]{Remark \ref*{rem}}.
After we find the lists of anchors in $G_a$ and $G_b$, we obtain a pair of induced bipartite subgraphs: \(G_a[V_a \setminus H_a,\wvec_a]\) and \(G_b[V_b \setminus H_b,\wvec_b]\).
When the anchor lists are the same, i.e. \(\wvec_a = \wvec_b\), \hyperref[alg:bip]{Algorithm \ref*{alg:bip}} can be applied, but bipartite graphs do not have the joint distribution $ER(n-h,h,\pvec)$, required for \hyperref[rem]{Remark \ref*{rem}}.
This is due to the fact that we used edge information to partition the original vertex set, so the edges are not independent of this partition.
However, this dependence is weak.
In \hyperref[subsec-general]{Section \ref*{subsec-general}} we will apply \hyperref[lemma:bip-error]{Lemma \ref*{lemma:bip-error}} after careful conditioning.
\bcomment{This will allow us to show that if the condition on $h$ in \hyperref[rem]{Remark \ref*{rem}} is satisfied, the alignment stage of \hyperref[alg]{Algorithm \ref*{alg}} does not make any mistakes.}


\subsection{General alignment algorithm}

\label{subsec-general}

In this section we first show that the anchor alignment stage is independent from the alignment of any pair of non-anchor vertices in the bipartite alignment step. We do this by considering the subgraph obtained by removing any pair of vertices and show that the anchor set is sufficiently stable due to the degree separation of at least 3 as guaranteed by \hyperref[theorem:HDMatching]{Theorem \ref*{theorem:HDMatching}}. This then allows us to combine results on both stages to get the condition for successful alignment of pairs of random graphs.

Recall that \(\wvec_a = f_{h}(G_a)\) and \(\wvec_b = f_{h}(G_b)\).
For \(U = \{u_1,u_2\} \subseteq V\), the induced bipartite subgraphs
\((G_a[U,\wvec_a],G_b[U,\wvec_b])\) determine whether Algorithm 1 misaligns \(u_1\) with \(u_2\) or \(u_2\) with \(u_1\).
However, these graphs do not have a correlated ER joint distribution, so we define a related pair of induced bipartite subgraphs.

\begin{definition}
  \label{def:U}
  Let \(G_a\) and \(G_b\) be graphs on vertex set \(V\).
  For set \(U = \{u_1,u_2\} \subseteq V\), and \(h\in\natS\), define
  \begin{align*}
    \wvec_a^{U} = f_{h}(G_a[V\setminus\{u_1,u_2\}]) \quad \textrm{ and } \quad B_a^{U} = G_a[U,\wvec_a^{U}],
  \end{align*}
  i.e. \((u,i)\in E(B_a^U) \iff \{u,w_i^U\} \in E(G_a)\) for any \(u\in U\) and \(i\in[h]\). Define \(\wvec_b^{U}\) and \(B_b^{U}\) analogously.
  Let \(\mcE^{\operatorname{H}}(U)\) be the event \(\wvec_a^{U} = \wvec_b^{U}\).
\end{definition}
We emphasize that in both \(B_a^{U}\) and \(B_b^{U}\) the left vertex set is \(\{u_1,u_2\}\) and the right vertex set is \([h]\), so the vertex sets are not random variables.

We start by stating a result on conditional independence of the high-degree neighborhoods of a pair of vertices.
\begin{lemma}
  \label{lemma:bip-dist}
  Let \((\Ga,\Gb)\sim ER(n;\pvec)\) be correlated graphs on the vertex set \(V\) and let \(U = \{u_1,u_2\} \subseteq V\).
  Then
  \begin{align*}
    B_a^{U} \sim ER(2,h,\plx), \quad B_b^{U} \sim ER(2,h,\pxl)
  \end{align*}
  \begin{align*}
      \textrm{and} \qquad (B_a^{U},B_b^{U})|\mcE^{\operatorname{H}}(U) \sim ER(2,h,\pvec),
  \end{align*}
  where \(B_a^U\) and \(B_b^U\) are as defined in \hyperref[def:U]{Definition \ref*{def:U}}
\end{lemma}
\begin{proof}
  Recall that, by definition, \(B_a^{U} = G_a[U,\wvec_a^{U}]\) and \(\ind{(u,j) \in E(B_a^{U})} = \ind{\{u,w_{a,j}^U\} \in E(G_a)}\).
  We will show that despite being defined using \(\wvec_a^{U}\), the random variable \(B_a^{U}\) is independent of the random variable \(\wvec_a^{U}\).
  Observe that \(B_a^{U} = G_a[U,\wvec_a^{U}]\) is independent of \(G_a[V\setminus U]\) because they have no edge random variables in common.
  Because \(\wvec_a^{U}=f_{h}(G_a[V\setminus U])\), \(B_a^{U}\) is independent of \(\wvec_a^{U}\) as well.

  Similarly, \(B_b^{U}\) is independent of \(\wvec_b^{U}\) and \(\ind{(u,j) \in E(B_b^{U})} = \ind{\{u,w_{b,j}\} \in E(G_b)}\).
  As long as \(\wvec_a^{U} = \wvec_b^{U}\) holds, \(\ind{(u,j) \in E(B_a^{U})}\) and \(\ind{(u,j) \in E(B_b^{U})}\) have the joint distribution of a pair of corresponding edges in the correlated Erdős-Rényi model.
\end{proof}

This result may be counterintuitive because we are selecting the right vertex set of \(B^{U}_a\) using high degree vertices, but there the edge density of \(B^{U}_a\) is the same as \(G_a\).
For a fixed \((u,j)\in U \times [h]\), the random variable \(\ind{(u,j) \in E(B_a^{U})}\) is not determined by any single edge random variable from \(G_a\), but is a mixture of \(\ind{\{u,v\} \in E(G_a)}\) over all \(v \in V \setminus U\) because \(\wvec_b^{U}\) is random.
It is helpful to compare with \(G_a[U^{\wvec_a},\wvec_a]\), where \(U^{\wvec_a} = \{u_1,u_2\}\) is a uniformly random subset of \(V \setminus H_a\).
This bipartite graph is not distributed as \(ER(n,\plx)\) because edges of \(G_a\) are slightly more likely to be sampled than non-edges.

Recall from \hyperref[HD-events]{Definition \ref*{HD-events}} that \(\mcE_a^{\operatorname{S}}\) is defined as the event that \(\degseq_{a,i} > \degseq_{a,i+1} + 2\) for all \(i \in [h]\) and $\mcE_b^{\operatorname{S}}$ is the corresponding event for $\wvec_b$ and $G_b$. 
\begin{lemma}
  \label{lemma:deg-spacing}
  The event $\mcE_a^{\operatorname{S}}$ implies \(\wvec_a = \wvec_a^{U}\) for all \(U \subseteq V\) pair of vertices that do not include any from \(\wvec_a\).
  Similarly $\mcE_b^{\operatorname{S}}$ implies \(\wvec_b = \wvec_b^{U}\).
\end{lemma}
\begin{proof}
  For any \(v \in V\), the degree of \(v\) in \(G_a\) differs by at most 2 from the degree of \(v\) in \(G_a[V\setminus U]\). The same holds for \(G_b\).
\end{proof}
Finally we prove our main theorem:

\bcomment{
\newtheorem*{Tmain}{Theorem \ref*{theorem:main}}
\begin{Tmain}
  \mainthm
\end{Tmain}
}

\begin{proof}[\underline{\textbf{Proof of Theorem \ref*{theorem:main}}}]
  \hyperref[theorem:HDMatching]{Theorem \ref*{theorem:HDMatching}} provides the condition on the correlation of graphs required to successfully align a given number \(h\) of high-degree vertices.
  From the inequalities \(h \leq \mathcal{O}\pth{\frac{\log n}{p_{11}}}\), \(\log h \leq \log n\), and the conditions in the theorem statement, \(\pll \geq \omega\pth{n^{-1/5}\log^{7/5} n}\) and \(\pol+\plo \leq o\pth{\frac{\pll ^5}{\log^6 n}}\), we have 
  \begin{align*}
    \max\acc{(\log h)^2, n(\plo+\pol)\log h} \leq o\pth{\frac{n\,\pll}{h^4\log n}\cdot\frac{\pll}{\plx}}.
  \end{align*}
  Thus \(P\croc{\mcE^{\operatorname{H}} \wedge \mcE_a^{\operatorname{S}} \wedge \mcE_b^{\operatorname{S}}} \geq 1-o(1)\), where \(\mcE^{\operatorname{H}}\), \(\mcE_a^{\operatorname{S}}\) and \(\mcE_b^{\operatorname{S}}\) are events as defined in \hyperref[HD-events]{Definition \ref*{HD-events}}.
  These events imply \(H_a = H_b = H\).

  Recall the definition of \(\mcE^{\operatorname{M}}(B_a,B_b)\) from \hyperref[lemma:bip-error]{Lemma \ref*{lemma:bip-error}}
   and \(\mcE^{\operatorname{H}}(U)\) from \hyperref[def:U]{Definition \ref*{def:U}}.
  Applying the union bound to error events in the bipartite alignment stage of the algorithm results in the following:
  \begin{align*}
    P[\widehat{M} &\neq M | \mcE^{\operatorname{H}} \wedge \mcE_a^{\operatorname{S}} \wedge \mcE_b^{\operatorname{S}}]\\
    &\leq \sum_{\{u_1,u_2\} \subseteq \rem} P\croc{\mcE^{\operatorname{M}}(G_a[U,\wvec_a],G_b[U,\wvec_b]) \wedge \mcE^{\operatorname{H}} \wedge \mcE_a^{\operatorname{S}} \wedge \mcE_b^{\operatorname{S}}}\\
    &\leql{a} \sum_{\{u_1,u_2\} \subseteq \rem} P\croc{\mcE^{\operatorname{M}}(B_a^{U},B_b^{U}) \wedge \mcE^{\operatorname{H}}(U)}\\
    &\leql{b} \sum_{\{u_1,u_2\} \subseteq V} P\condr{\mcE^{\operatorname{M}}(B_a^{U},B_b^{U})}{\mcE^{\operatorname{H}}(U)}\\
    &\leql{c} \sum_{\{u_1,u_2\} \subseteq V} \exp(-h \rho^2).
  \end{align*}
  The inequality $(a)$ is derived by applying \hyperref[lemma:deg-spacing]{Lemma \ref*{lemma:deg-spacing}} twice, which gives \(G_a[U,\wvec_a] = B_a^{U}\), \(G_a[U,\wvec_b] = B_b^{U}\), and \(\wvec_a^{U} = \wvec_b^{U}\).
  (Recall that the event \(\{\wvec_a^{U} = \wvec_b^{U}\}\) is denoted by \(\mcE^{\operatorname{H}}(U)\).)
  In $(b)$, we use \(P[\mcE^{U}] \leq 1\) and also extend the sum to include pairs \(\{u_1,u_2\}\) that include members of \(H\).
  Because \(u_1\) and \(u_2\) are now arbitrary vertices with no conditioning, from \hyperref[lemma:bip-dist]{Lemma \ref*{lemma:bip-dist}} we have that \((B_a^{U},B_b^{U}) \sim ER(2,h,\pvec)\).
  Observe that for any \(U = \{u_1,u_2\} \subseteq \rem\), the signatures in \hyperref[lemma:bip-error]{Lemma \ref*{lemma:bip-error}} are the same as the signatures in Algorithm 1: \(\Sx{G_a[U,\wvec_a]}{u_i} = \SBa{u_i}\).
  Finally, $(c)$ follows from \hyperref[lemma:bip-error]{Lemma \ref*{lemma:bip-error}}.
  Note that the final bound is the same as the one stated earlier in \hyperref[rem]{Remark \ref*{rem}}.
  
  We have
  \begin{align*}
    \rho
    &= \sqrt{\poo\plx + \pll \pox}-\sqrt{\plo\pox + \pol\plx}\\
    &= \sqrt{\pll\poo}\pth{\sqrt{2 + \frac{\plo}{\pll} + \frac{\pol}{\poo}}-\sqrt{2 \frac{\pol\plo}{\pll\poo} + \frac{\plo}{\pll} + \frac{\pol}{\poo}}}\\
    &\geq \sqrt{2\pll}\pth{1-\mathcal{O}\pth{\frac{1}{\log n}}}
  \end{align*}
  because \(\frac{\plo}{\pll} \leq o\pth{\frac{1}{\log^6 n}}\) and \(\frac{\pol}{\poo} \leq o\pth{\frac{1}{\log^6 n}}\).
  The logarithm of the probability of an incorrect alignment in $\rem$ is at most
  \begin{align*}
    \log&\pth{n(n-1)\exp(-h \rho^2)}\\
    &\leq 2 \log n - \frac{\log n + \omega(1)}{p_{11}}\cdot2 p_{11}\pth{1-\mathcal{O}\pth{\frac{1}{\log n}}}\\
    &\leq 2 \log n - 2\log n + \mathcal{O}(1) - \omega(1)\pth{1-\mathcal{O}\pth{\frac{1}{\log n}}}
    = -\omega(1).
  \end{align*}
\end{proof}



\section{Implementation and Performance Evaluation}
\label{simulations}

In this section, we study the performance of our canonical labeling algorithm through simulations over real and synthetic data. In \hyperref[subsec:implementation]{Section \ref*{subsec:implementation}}, we  describe our slight modification to the original algorithm to improve its performance in small graphs. In \hyperref[subsec:simER]{Section \ref*{subsec:simER}} and \ref{subsec:simProtein},  we compare the performance of our algorithm against  EigenAlign and LowRankAlign \cite{feizi2016spectral} on synthetically generated correlated ER graphs and on a protein network, respectively.

\subsection{Implementation}
\label{subsec:implementation}

We consider an implementation of a variant of \hyperref[alg]{Algorithm \ref*{alg}} for finite graphs. The modifications introduced over the original algorithm provide some robustness against certain events that have low likelihood in the asymptotic case but that become more significant when considering small graphs. Thus our theoretical analysis applies equally to the modified algorithm.

\textbf{Consistent bipartite alignment:} By \hyperref[lemma:bip-error]{Lemma \ref*{lemma:bip-error}}, with high probability within the regime of interest, there is a unique signature \(\Sb{v}\) in \(\Gb\) at minimum distance to any signature \(\Sa{u}\) in \(\Ga\). If there are two signatures from \(\Gb\) that both lie at minimum distance to a given signature in \(\Ga\), the naive approach considered in the analysis would simply align both vertices from \(\Vb\) to the same vertex in \(\Va\). We impose a requirement of `consistency' to the signature alignment operation that acts less naively in this event. Let \(D\in\natS^{U_a\times U_b}\) be the matrix whose entries \(D_{u,v}\) correspond to the Hamming distance between signatures \(\Sa{u}\) and \(\Sb{v}\) obtained by anchor list \(H_a\) and \(H_b\). Let \(\mu_{a\to b}:U_a\to U_b\) and \(\mu_{b\to a}:U_b \to U_a\) denote the position of the minimum value in any row or column respectively. Consistent signature alignment aligns \((u,v)\) is aligned if and only if all of the following hold: \(\forall u' \in U_a\setminus\{u\}\), \(\mu_a(u) \neq v\), \(\forall v' \in U_b\setminus\{v\}\), \(\mu_b(v) \neq u\) and either \(\mu_a(u) = v\) or \(\mu_b(v) = u\). Consistent signature alignment might leave some vertices unmatched, in which case we perform another alignment until all vertices have been matched.

\textbf{Robust anchor alignment:} By \hyperref[corollary:HDMatchingN]{Corollary \ref*{corollary:HDMatchingN}} and \hyperref[theorem:Bollobas]{Theorem \ref*{theorem:Bollobas}}, the degree sequence on the higher extreme is well separated in \(\Ga\) which guarantees it preserving the same order in \(\Gb\). The same argument can be shown to apply for the lower extreme of the degree sequence. Thus in the implementation we extract anchors from both extremes. Furthermore we consider a modification that filters out anchors that appear to be misaligned. This is done as follows: We pick a given number of vertices from both extremes of the degree sequence in both graphs and align them one-by-one according to their position in the degree sequence. Then, using this alignment as anchors, we perform consistent signature alignment over the same subset of vertices to get a new alignment. Then we construct the agreement graph \(G_{\textrm{agr}}\), i.e. a graph over the aligned pairs of vertices where any edge \(e\in G_{\textrm{agr}}\) if and only if \(e \in \Ga\) and \(e\in \Gb\) or \(e\notin \Ga\) and \(e \notin \Gb\)). We prune this graph down to a minimum size and consider the surviving pairs of aligned vertices to be our anchors. We iteratively repeat this process of degree alignment - signature alignment - pruning. At each iteration, degree alignment is only performed on the vertices that haven't been included in the final anchor set in the previous iteration. We stop iterating when the pruned agreement graph's density stops increasing between iterations.

Note that neither modification changes the performance of the algorithm as \(n\to\infty\) since the events where the original algorithm would give a different outcome than the variant has occur with probability \(o(1)\) in the regime of interest.

\subsection{Performance over Erdős-Rényi graphs}
\label{subsec:simER}

We ran simulations on correlated ER graphs of size ranging from \(n=128\) to \(n=16,384\) to see how well our theoretical results generalize to small graphs. As expected, the algorithm’s performance suffers  in small graphs as a result of the discreteness of degrees. However this effect becomes less significant as the graphs grow in size.

We ran the algorithm over 20 pairs of correlated random Erdős-Rényi graphs with \(\pll = 1/4\) for various values of \(\plo=\pol\). We report the noise level as \(-\frac{\log_2 \plo}{\log_2 n}\) which is the relevant measure as seen in \hyperref[regions]{Fig. \ref*{regions}}. In \hyperref[fig:mean]{Fig. \ref*{fig:mean}} we give the mean performance of these experiments while \hyperref[tab:median]{Table \ref*{tab:median}} shows the median over the experiments. The performance of the algorithm increases as we consider larger graphs. 
We also note that, as seen in \hyperref[tab:median]{Table \ref*{tab:median}}, our implementation tends to either properly align nearly all vertices or almost none of them.

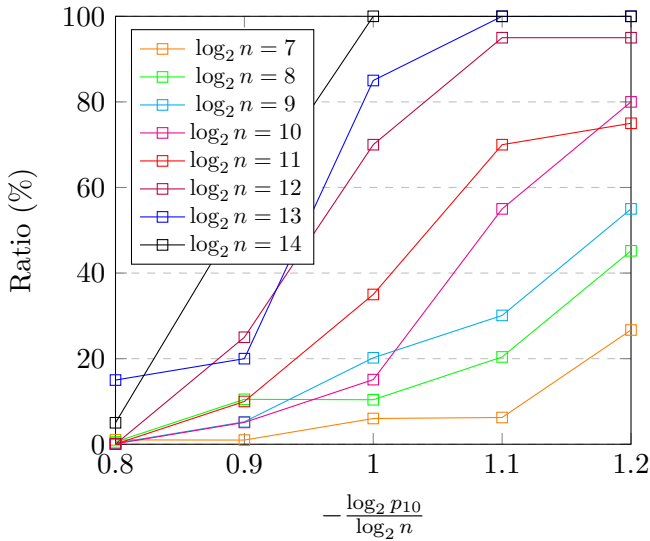
\begin{figure}[h]
\centering
\caption{Ratio of properly aligned vertices (mean over 20 random graph pairs)}
\label{fig:mean}
\begin{tikzpicture}
\begin{axis}[
    xlabel={\(-\frac{\log_2 \plo}{\log_2 n}\)},
    ylabel={Ratio (\%)},
    xmin=0.8, xmax=1.2,
    ymin=0, ymax=100,
    xtick={0.8, 0.9, 1, 1.1, 1.2},
    ytick={0,20,40,60,80,100,120},
    legend pos=north west,
    ymajorgrids=true,
    grid style=dashed,
    legend style={nodes={scale=0.8, transform shape}}
]
 
\addplot[
    color=orange,
    mark=square,
    ]
    coordinates {
    (0.8, 1.02)(0.9,0.977)(1,6.02)(1.1,6.25)(1.2,26.7)
    };
    \addlegendentry{\(\log_2 n = 7\)};

\addplot[
    color=green,
    mark=square,
    ]
    coordinates {
    (0.8,0.547)(0.9,10.5)(1,10.4)(1.1,20.4)(1.2,45.2)
    };
    \addlegendentry{\(\log_2 n = 8\)};
    
\addplot[
    color=cyan,
    mark=square,
    ]
    coordinates {
    (0.8,0.225)(0.9,5.25)(1,20.2)(1.1,30.1)(1.2,55.0)
    };
    \addlegendentry{\(\log_2 n = 9\)};

\addplot[
    color=magenta,
    mark=square,
    ]
    coordinates {
    (0.8,0.102)(0.9,5.10)(1,15.1)(1.1,55.0)(1.2,80.0)
    };
    \addlegendentry{\(\log_2 n = 10\)};
    
\addplot[
    color=red,
    mark=square,
    ]
    coordinates {
    (0.8,0.0439)(0.9,10.0)(1,35.0)(1.1,70.0)(1.2,75.0)
    };
    \addlegendentry{\(\log_2 n = 11\)};

\addplot[
    color=purple,
    mark=square,
    ]
    coordinates {
    (0.8,0.0256)(0.9,25.0)(1,70.0)(1.1,95.0)(1.2,95.0)
    };
    \addlegendentry{\(\log_2 n = 12\)};

\addplot[
    color=blue,
    mark=square,
    ]
    coordinates {
    (0.8,15.0)(0.9,20.0)(1,85.0)(1.1,100)(1.2,100)
    };
    \addlegendentry{\(\log_2 n = 13\)};

\addplot[
    color=black,
    mark=square,
    ]
    coordinates {
    (0.8,5.02)(0.9,55.2)(1,100)(1.1,100)(1.2,100)
    };
    \addlegendentry{\(\log_2 n = 14\)};
 
\end{axis}
\end{tikzpicture}
\end{figure}

\begin{table}[h]
    \centering
    \begin{tabular}{l|c c c c c|}
        \cline{2-6}
        & \multicolumn{5}{|c|}{\(-\log_2 \plo / \log_2 n\)}\\
         & 0.8 & 0.9 & 1.0 & 1.1 & 1.2 \\
        \hline
        \(\log_2 n = 7\) & 0.78 & 0.78 & 1.17 & 1.56 & 1.95 \\
        \hline
        \(\log_2 n = 8\) & 0.39 & 0.39 & 0.39 & 0.39 & 1.17 \\
        \hline
        \(\log_2 n = 9\) & 0.20 & 0.20 & 0.20 & \textcolor{red}{0.20} & \textcolor{green}{99.61} \\
        \hline
        \(\log_2 n = 10\) & 0.10 & 0.10 & \textcolor{red}{0.10}  & \textcolor{green}{100.00} & 99.90 \\
        \hline
        \(\log_2 n = 11\) & 0.05 & 0.05 & \textcolor{red}{0.10} & \textcolor{green}{100.00} & 100.00 \\
        \hline
        \(\log_2 n = 12\) & 0.02 & \textcolor{red}{0.02} & \textcolor{green}{100.00} & 100.00 & 100.00 \\
        \hline
        \(\log_2 n = 13\) & 0.01 & \textcolor{red}{0.02} & \textcolor{green}{100.00} & 100.00 & 100.00 \\
        \hline
        \(\log_2 n = 14\) & \textcolor{red}{0.01} & \textcolor{green}{99.98} & 100.00 & 100.00 & 100.00\\
        \hline
    \end{tabular}
    \caption{Ratio of properly aligned vertices (median over 20 random graph pairs)}
    \label{tab:median}
\end{table}

We also compared the algorithm’s performance with EigenAlign and LowRankAlign \cite{feizi2016spectral}. Since these algorithms do not scale well for large graphs, we only considered a small graph of size \(n=128\). This setting is unfavorable for our canonical labeling algorithm (since anchor alignment is difficult due to the discrete nature of degrees). Yet we still observe that the algorithm outperforms EigenAlign and is comparable to LowRankAlign for low noise.


Algorithms based on a semidefinite programming relaxation of quadratic alignment have also been proposed \cite{feizi2016spectral}, but are not computationally feasible even for $n=128$.

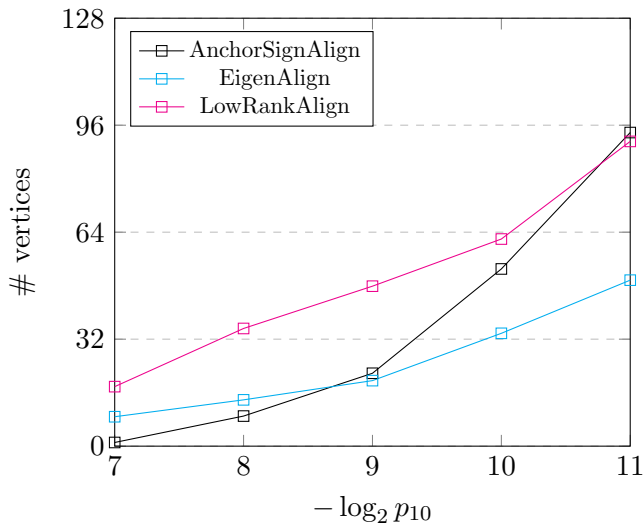
\begin{figure}[h]
\centering
\caption{Number of properly aligned vertices (mean over 20 random graph pairs)}
\label{fig:meanCompare}
\begin{tikzpicture}
\begin{axis}[
    xlabel={\(-\log_2 \plo\)},
    ylabel={\# vertices},
    xmin=7, xmax=11,
    ymin=0, ymax=128,
    xtick={7, 8, 9, 10, 11},
    ytick={0,32,64,96,128},
    legend pos=north west,
    ymajorgrids=true,
    grid style=dashed,
    legend style={nodes={scale=0.8, transform shape}}
]
 
\addplot[
    color=black,
    mark=square,
    ]
    coordinates {
    (7,1.1)(8,9.0)(9,21.8)(10,53.0)(11,93.8)
    };
    \addlegendentry{AnchorSignAlign};

\addplot[
    color=cyan,
    mark=square,
    ]
    coordinates {
    (7,8.8)(8,13.85)(9,19.6)(10,33.75)(11,49.65)
    };
    \addlegendentry{EigenAlign};
    
\addplot[
    color=magenta,
    mark=square,
    ]
    coordinates {
    (7,17.8)(8,35.2)(9,47.85)(10,61.95)(11,91.15)
    };
    \addlegendentry{LowRankAlign};
 
\end{axis}
\end{tikzpicture}
\end{figure}

\bcomment{
\begin{table}[h]
    \centering
    \begin{tabular}{l|c c c c c|}
        \cline{2-6}
        & \multicolumn{5}{|c|}{\(-\log_2 \plo\)}\\
         & 7 & 8 & 9 & 10 & 11 \\
        \hline
        AnchorSignAlign & 1 & 1 & 3 & 5 & 128 \\
        \hline
        EigenAlign & 9 & 14 & 20 & 33 & 49 \\
        \hline
        LowRankAlign & 20 & 40 & 51 & 66 & 98 \\
        \hline
    \end{tabular}
    \caption{Number properly aligned vertices (median over 20 random graph pairs)}
    \label{tab:medianCompare}
\end{table}
}

\subsection{Simulation over protein-protein interaction network}
\label{subsec:simProtein}

To  study the performance of the algorithm in actual networks, we ran simulations on a protein-protein interaction network. The distribution of such networks often is quite different from the ER model. Our results show that the algorithm is applicable as long as noise level is low enough.

The implementation has been run over the protein-protein interaction network of Campylobacter jejuni, which a species commonly considered when studying cross-species alignments of protein networks \cite{saraph2014}. Since AnchorSignAlign is only suitable for the alignment of graphs whose vertex sets have a one-to-one correspondence, we generate a pair of correlated graphs from this single network by subsampling at various rates \(s\). (The probability of any edge from the original graph being included in any of the new graphs is \(s\) independently from all other edges and the other new graph.)

\begin{table}[h]
    \centering
    \begin{tabular}{|l|c c c c c c c|}
        \hline
         \(-\log_2 (1-s)\) & 5 & 6 & 7 & 8 & 9 & 10 & \(\infty\) \\
        \hline
        \# correctly aligned & 10 & 751 & 755 & 763 & 765 & 765 & 765\\
        \hline
    \end{tabular}
    \caption{Number of vertices properly aligned by AnchorSignAlign over a pair of subsampled networks with different subsampling rates (n = 1095)}
\end{table}

The algorithm shows robustness against noise up to \(2^{-6}\) over this network. While the performance appears to plateau once the noise level goes below that value, this is in fact due to the automorphisms of the network, as many proteins are not distinguishable from others by simply considering the structure of the protein-protein interaction network. 

We have not been able to test EigenAlign and LowRankAlign on any protein-protein interaction network as these typically have more than 1000 nodes.

\subsection{Computational time}

Experimental results show the run-time of our implementation to scale as \(t \approx 0.5s \times (n^2\log_2 n)\).

\begin{table}[h]
    \centering
    \begin{tabular}{|l|c c c c|}
        \hline
         \(\log_2 n\) & 11 & 12 & 13 & 14 \\
        \hline
        \(t/(n^2\log_2 n)\) & 0.52 & 0.52 & 0.53 & 0.41\\
        \hline
    \end{tabular}
    \caption{Scaling factor for different values of \(n\)}
    \label{tab:medianCompare}
\end{table}

This is significantly better than the run-time of EigenAlign and LowRankAlign. This shows this approach to be suitable to perform alignment over very large graphs.

\begin{figure}[h]
\centering
\caption{Average time to compute an alignment}
\label{fig:timeCompare}
\begin{tikzpicture}
\begin{axis}[
    ymode = log,
    xlabel={\(\log_2 n\)},
    ylabel={time (s)},
    xmin=4.5, xmax=14,
    ymin=0.1, ymax=1000,
    xtick={5, 6, 7, 8, 9, 10, 11, 12, 13, 14},
    ytick={0.1, 1,10,100,1000},
    legend pos=south east,
    ymajorgrids=true,
    grid style=dashed,
    legend style={nodes={scale=0.8, transform shape}}
    ]
]
 
\addplot[
    color=black,
    mark=square,
    ]
    coordinates {
    (7,0.1)(8,0.2)(9,0.3)(10,0.8)(11,2.4)(12,10.7)(13,46.3)(14,156)
    };
    \addlegendentry{AnchorSignAlign};

\addplot[
    color=cyan,
    mark=square,
    ]
    coordinates {
    (5,1.1)(6,6.7)(7,93)
    };
    \addlegendentry{EigenAlign};
    
\addplot[
    color=magenta,
    mark=square,
    ]
    coordinates {
    (5,0.38)(6,4.1)(7,53)(8,890)
    };
    \addlegendentry{LowRankAlign};
 
\end{axis}
\end{tikzpicture}
\end{figure}
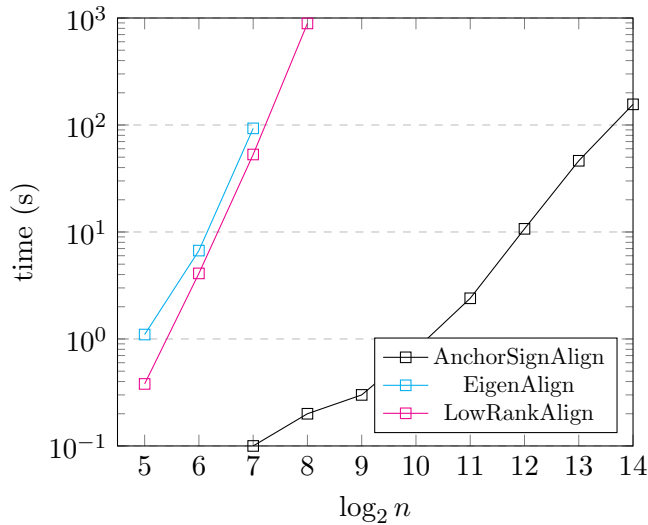


\section{Conclusion}
\label{conclusion}

We studied the performance of a canonical graph matching algorithm under the correlated ER graph model and obtained the expression for the region where the algorithm succeeds. To do so we analyzed the two steps of the algorithm, comprised of a high-degree matching and a subsequent bipartite matching. The first step which identified the pairing between high-degree subset of vertices can provide an initial set of seed vertices which may be used for various seed-based matching approaches. In this work we used a particular bipartite matching algorithm based on signatures derived from connections of the remaining (i.e. unidentified) vertices to the high-degree vertices identified at the first step.

There are a number of possible directions in which this work can be extended. One would be removing the assumption of the current model that the two vertex sets have a one-to-one correspondence. This would allow the analysis of more realistic scenarios where both graphs can potentially contain many vertices that have no exact match in the other. In this case it is necessary to avoid matching such vertices by considering some measure of the strength of correspondence between matching candidates. Another direction to consider is the scenario where information offered by the graph structure is richer. This would be the case when the edges are directed or weighted. It could also be the case that adjacency relations are defined by more than 2 states, rather than our model where the existence and absence of edges are the only 2 states. Another model with richer information would be the case of hypergraphs. Finally we note that some extensions to the algorithm, (such as considering highest-degree vertices at distance two rather than only immediate neighbor in the bipartite matching step) could provide considerable improvements to the region where the algorithm succeeds.

\bcomment{
\blue{***To be rewritten:***}
One direction in which this work could be extended is the analysis of other algorithms that could handle more complex scenarios. The considered algorithm has important limitations. It only applies to pairs of graphs whose vertex sets have an exact one-to-one correspondance, therefore it does not have use in the realistic case where the two networks significantly differ in their set of nodes. Another scenario would be the case where connections in the network contain more information, such as directions and/or weights which would be represented as directed and/or weighted graphs respectively. A further elaboration would consider connections involving multiple nodes, corresponding to hyperedges in graphs.
}

\bcomment{
Finally we must recognize the limitations of our model as ER graphs do not present a perfect match for many real life networks. Therefore evaluating the algorithm under different models could be of interest as it might provide more insight related to real-life applications.
}


\section*{Acknowledgments}
This work was supported in part by NSF grant CCF 16-17286 and MURI grant W911NF-15-1-0479.

\appendix


\bcomment{
\section{Construction of \hyperref[regions]{Fig. \ref*{regions}}}
\label{appendix:regions}
}

\section{Proofs of lemmas on anchor alignment}

\begin{proof}[\underline{\textbf{Proof of Lemma \ref{lemma:HDMatching2}}}]
  Let us denote the degree separations in the two graphs by \(\alpha \define \da{u}-\da{v}\) and \(\beta \define \db{u}-\db{v}\).
  Observe that the presence of the edge \(\{u,v\}\) in \(G_a\) does not affect \(\alpha\).
  Thus we define
  \begin{align*}
      \daprime{u} \define |N_a(u) \setminus \{v\}| \qquad &\dxaprime{u} \define n-2-\daprime{u}\\
      \daprime{v} \define |N_a(v) \setminus \{u\}| \qquad &\dxaprime{v} \define n-2-\daprime{v}.
  \end{align*}
  The error event in the degree sequence, i.e. \(\db{u}-\db{v}\leq k\), corresponds to \(\beta\leq k\). By the Chernoff bound:
  \begin{align*}
    \prob{\beta\leq k|\daprime{u},\daprime{v}} \leq z^{-k}\E{z^{\beta}|\daprime{u},\daprime{v}} \phantom{5} \forall 0 < z \leq 1.
  \end{align*}
  In \hyperref[appendix:beta]{Appendix \ref*{appendix:beta}} we derive an expression for the probability generating function \(F_\beta(z) \define \E{z^{\beta}|\daprime{u},\daprime{v}}\):
  \begin{align*}
    F_\beta(z) = &z^{\alpha}\pth{1+\plox(z-1)}^{\daprime{v}}\pth{1+\polx(z-1)}^{\dxaprime{u}}\\
    &\times\pth{1+\plox\pth{z^{-1}-1}}^{\daprime{u}}\pth{1+\polx\pth{z^{-1}-1}}^{\dxaprime{v}}.
  \end{align*}
  
  By applying \(1+x\leq e^x\) we get
  \begin{align}
      F_\beta(z) \leq &\exp\acc{\alpha\log z + \pth{\plox \daprime{v} + \polx \dxaprime{u}}(z-1)}\nonumber\\
      &\times\exp\acc{\pth{\plox \daprime{u} + \polx \dxaprime{v}}\pth{z^{-1}-1}}\nonumber
      \intertext{Furthermore applying \(\log x \leq x-1\) we have}
      z^{-k}F_\beta(z) \leq &\exp\acc{\pth{\alpha -k +\plox \daprime{v} + \polx \dxaprime{u}}(z-1)}\nonumber\\
      &\times\exp\acc{\pth{\plox \daprime{u} + \polx \dxaprime{v}}\pth{z^{-1}-1}} \label{ineq:lemma:HDMatching2}
  \end{align}
  Denote the coefficients by
  \begin{align*}
      r' \define \alpha - k+\plox \daprime{v} + \polx \dxaprime{u} \quad \textrm{ and } \quad r \define \plox \daprime{u} + \polx \dxaprime{v}.
  \end{align*} Denote their difference as
  \begin{align*}
      \Delta r \define r' - r &= \alpha - k + \plox \daprime{v} - \plox\Bpth{(n-2)-\daprime{v}}\\
      &\qquad - \plox\daprime{u} + \polx\Bpth{(n-2)-\daprime{u}}\\
      &= \alpha - k - \pth{\polx + \plox}(\daprime{u}-\daprime{v})\\
      &= \alpha \pth{1-\polx-\plox} - k
  \end{align*}\
  The right hand side of the inequality in \hyperref[ineq:lemma:HDMatching2]{(\ref*{ineq:lemma:HDMatching2})} is minimized at \(z^* \define \sqrt{r/r'}\). Taking the logarithm of both sides in \hyperref[ineq:lemma:HDMatching2]{(\ref*{ineq:lemma:HDMatching2})} and evaluating it at \(z=z^*\) we get
  \begin{align*}
      \log F_\beta(z^*)-k\log z^* \leq -\pth{\sqrt{r'}-\sqrt{r}}^2 = -\Delta r \pth{\sqrt{1+r/\Delta r}-\sqrt{r/\Delta r}}^2.
  \end{align*}
  The inequality \(\sqrt{1+x^2}-x\geq \frac{1}{1+2x}\) holds for any \(x\geq 0\). Specifically the choice of \(x = \sqrt{r/\Delta r}\) results in
  \begin{align*}
      -\croc{\log F_\beta(z^*)-k\log z^*} \geq \frac{\Delta r}{(1+2\sqrt{r/\Delta r})}
  \end{align*}
  Note that:
  \begin{align*}
    \Delta r \geq 4\max\acc{\eh,\sqrt{r\eh}} 
    \implies \pth{1+2\sqrt{\frac{r}{\Delta r}}}^2 &\leq 4\min\acc{1,\frac{r}{\eh},\sqrt{\frac{r}{\eh}}}\\
    \implies \frac{\Delta r}{\pth{1+2\sqrt{\frac{r}{\Delta r}}}^2} &\geq \frac{1}{4}\max\acc{\Delta r',\Delta r\frac{\eh}{r},\Delta r\sqrt{\frac{\eh}{r}}}\\
    &\geq\frac{\Delta r}{4}\geq\eh
  \end{align*}
  which implies \(\pth{z^*}^{-k}F_\beta(z^*) \leq e^{-\eh}\). Finally observe that \(\varphi \define \da{u}\plox + \dxa{v}\polx\) is at least \(r\). Therefore the condition in the statement of the lemma implies \(\Delta r \geq 4\max\acc{\eh,\sqrt{r\eh}}\).
\end{proof}

\begin{proof}[\underline{\textbf{Proof of Lemma \ref{corollary:HDMatchingN}}}]
  Let \(H_a\) and \(S_a\) be the set of \(h\) and \(s\) highest-degree vertices in \(\Ga\) respectively and define \(H_b\) analogously for \(G_b\).
  The following two events collectively imply \(f_{h}(\Ga)=f_{h}(\Gb)\) and \(\degseq_{b,i} - \degseq_{b,i+1} > k\) for any \(i\in[h]\).
  \begin{itemize}
  \item Let \(\mcE^{\text{high}}\) be the event that vertices in \(H_a\)  have the same degree ordering in \(\Ga\) and in \(\Gb\) as well as a minimum degree separation larger than \(k\) in \(G_b\). Note that this does not guarantee \(H_a = H_b\).
  \item Let \(\mcE^{\text{low}}\) be the event that all vertices in \(V\diffS H_a\) have degree less than \(\degseq_{b,h}-k\) in \(\Gb\), i.e. no vertex from \(V\diffS H_a\) is in \(H_b\) and all have a sufficiently large degree separation with the \(h\)-th highest-degree vertex.
  \end{itemize}


  First we consider \(\mcE^{\text{high}}\), i.e. the event where \(\degseq_{b,i}-\degseq_{b,j}>k\) for any \(i < j\) with \(i,j\in[h]\).
  Notice that it is sufficient to check this condition for consecutive pairs of vertices in the degree sequence.
  Given the condition in \hyperref[corollary:HDMatchingN:cond1]{(\ref*{corollary:HDMatchingN:cond1})}, \hyperref[lemma:HDMatching2]{Lemma \ref*{lemma:HDMatching2}} states that for any pair of vertices \(v_i,v_{i+1}\in H_a\), \(v_i\) and \(v_{i+1}\) in \(\Gb\) have the same degree ordering as well as a degree separation larger than \(k\) with probability at least \(e^{-\eh}\).
  Thus, by the union bound, we get \(\prob{\overline{\mcE^{\text{high}}}} \leq 1-h e^{-\eh}\).
  
  Second we consider \(\mcE^{\text{low}}\), i.e. the event where \(\degseq_{b,h}-\degseq_{b,i}>k\) for any \(i\in [n]\diffS[h]\).
  By the condition in \hyperref[corollary:HDMatchingN:cond1]{(\ref*{corollary:HDMatchingN:cond1})} we have, \(\forall i\in [s] \setminus [h]\),
  \begin{align*}
    \degseq_{a,h} - \degseq_{a,{i}} &\geq (i-h)(k+4\max\acc{\eh,\sqrt{\varphi\cdot\eh}})(1-\varepsilon)^{-1}\\
    &\geq \pth{k+4\max\acc{(i-h)\eh,\sqrt{(i-h)\varphi\eh}}}(1-\varepsilon)^{-1}
    \intertext{and \(\forall i\in [n] \setminus [s]\),}
    \degseq_{a,h} - \degseq_{a,{i}} &\geq (s+1-h)(k+4t\max\acc{\eh,\sqrt{\varphi\cdot\eh}})(1-\varepsilon)^{-1}\\
    &\geq \pth{k+4\max\acc{(s+1-h)\eh,\sqrt{(s+1-h)\varphi\eh}}}(1-\varepsilon)^{-1}.
  \end{align*}
  By \hyperref[lemma:HDMatching2]{Lemma \ref*{lemma:HDMatching2}} we then have
  \begin{align*}
      \prob{\degseq_{a,h} - \degseq_{a,{i}} \leq k} \leq \exp(-\eh \min\{i-h,s+1-h\}).
  \end{align*}
  Then, by the union bound, 
  \begin{align*}
    \prob{\overline{\mcE^{\text{low}}}}
    &\leq \sum_{i=h+1}^{s}e^{-\eh(i-h)} + \sum_{i=s+1}^n e^{-\eh(s+1-h)}\\
    &\leq \frac{e^{-\eh}}{1-e^{-\eh}} + (n-s)\frac{h}{n}e^{-\eh}
  \end{align*}
  Applying the union bound again we obtain
  \[
    \prob{\overline{\mcE^{\text{high}}} \vee \overline{\mcE^{\text{low}}}} \leq 
    (2h+1)e^{-\eh}/(1-e^{-\eh}).
    \]
\end{proof}

\begin{proof}[\underline{\textbf{Proof of Lemma \ref{lemma:maximumDegree}}}]
  For any vertex \(u\in V(G)\), \(d_G(u)\sim\textrm{Bin}(n-1;p)\).
  By the Chernoff bound, for any \(D\in \natS\) and \(z\in[1,\infty]\)
  \begin{align*}
    \prob{d_G(u)\geq D} \leq z^{-D}\E{z^{-d_G(u)}} \leq z^{-D}\croc{1+p(z-1)}^{n-1}.
  \end{align*}
  Applying \(1+x\leq e^x\) to both terms this becomes
  \begin{align*}
      \log\prob{d_G(u)\geq D} \leq D(z^{-1}-1)+ p(n-1)(z-1).
  \end{align*}
  The right hand side is minimized for \(z^* = \sqrt{\frac{D}{p(n-1)}}\) which gives us 
  \begin{align*}
      \log \prob{d_G(u)\geq D}\leq -\pth{\sqrt{D}-\sqrt{p(n-1)}}^2.
  \end{align*}
  
  Let \(D = (1+\epsilon)p(n-1)\).
  By the union bound, the probability that the maximum degree is at least \(D\) is at most
  \begin{align*}
      n\prob{d_G(u)\geq D} &\leq n\exp\pth{-p(n-1)\pth{\sqrt{1+\epsilon}-1}^2}\\
      &\leq n\exp\pth{-\omega(\log n)} \leq o(1).
  \end{align*}
\end{proof}

\section{Proofs of Lemmas on Bipartite Alignment}

\begin{proof}[\underline{\textbf{Proof of Lemma \ref{lemma:bip-error}}}]
  Define the random variable 
  \begin{align*}
      \gamma = \norm{\SBa{v} - \SBb{u}} - \norm{\SBa{u}-\SBb{u}}.
  \end{align*}
  
  We bound the probability of \(\gamma \leq 0\) using the Chernoff bound: \(\prob{\gamma\leq 0 } \leq \E{z^{\gamma}}\) for all \(0 < z \leq 1\). The generating function \(F_{\gamma}(z) \define \E{z^{\gamma}}\) is given as
  \begin{align*}
      F_{\gamma}(z) = \croc{1+\qo(z-1)+\ql(z^{-1}-1)}^h
  \end{align*}
  where \(\qo = \poo\plx + \pll \pox\) and \(\ql = \plo\pox + \pol\plx\). (See \hyperref[appendix:gamma]{Appendix \ref*{appendix:gamma}} for derivation.)
  
  Applying \(1+x\leq e^x\) and evaluating the function at \(z^* = \sqrt{\frac{\ql}{\qo}}\), we get \(\log F_{\gamma}(z^*) \leq -h\pth{\sqrt{\qo}-\sqrt{\ql}}^2\).
  Hence for \(\rho = \sqrt{\qo}-\sqrt{\ql}\) we have \(\prob{\gamma\leq 0} \leq \exp\pth{-h\rho^2}\).
  
  Notice that for the analogous 
  \begin{align*}
      \gamma' = \norm{\SBa{u} - \SBb{v}} - \norm{\SBa{v}-\SBb{v}}
  \end{align*}
  the same bound holds. The event \(\mcE^{\operatorname{M}}(B_a,B_b)\) is equivalent to\\ \(\{\gamma\leq 0 \vee \gamma'\leq 0\}\). Thus by the union bound \(\prob{\mcE^{\operatorname{M}}(B_a,B_b)} \leq 2\exp\pth{-h\rho^2}\).
\end{proof}

\section{Derivations of probability generating functions}

\textbf{Probability generating function of the degree separation beta}
\label{appendix:beta}

Given the random graphs \((\Ga,\Gb) \sim ER(n;\pvec)\) and for a given pair of vertices \(u, v\), define \(\daprime{u} = |N_a(u) \setminus \{v\}|\), \(\daprime{v} = |N_a(v) \setminus \{u\}|\) and \(\dbprime{u}\),\(\dbprime{v}\) analogous for \(\Gb\). We seek to find \(F_{\beta}(z) = \E{z^{\beta}|\daprime{u},\daprime{v}}\) where \(\beta = \db{u} - \db{v}\).

Let us denote the degree separation in \(\Ga\) as \(\alpha = \da{u} - \da{v}\). Note that \(\daprime{u}-\daprime{v} = \da{u} - \da{v} = \alpha\) and \(\dbprime{u}-\dbprime{v} = \db{u} - \db{v} = \beta\). Let us denote the number edges of \(x\) in \(\Ga\setminus\{u\}\) that are non-edges in \(\Gb\) (i.e. number of edges \textit{exclusive} to \(\Ga\)) as \(e_a^u = \norm{\Na{u} \setminus \Nb{u} \setminus\{v\}}\) and vice versa as \(e_b^u = \norm{\Nb{u} \setminus \Na{u} \setminus\{v\}}\). It can be shown that \(\dbprime{u} = \daprime{u} - e_a^u + e_b^u\). Similarly define \(e_a^v\) and \(e_b^v\) by ignoring the edge \(\{u,v\}\). We then have
\begin{align*}
    \dbprime{u}-\dbprime{v} &= \daprime{u}-\daprime{v} - e_a^u + e_a^v + e_b^u - e_b^v
\end{align*}
or simply \(\beta = \alpha - e_a^u + e_a^v + e_b^u - e_b^v\).
Notice that given \(\daprime{u}\) and \(\daprime{v}\), \(\alpha\) is deterministic. Also notice that the remaining terms \(e_a^u,e_a^v,e_b^u,e_b^v\) are mutually independent binomially distributed random variables with distribution:
\begin{align*}
    &e_a^u \sim B\pth{\daprime{u};\plox}, \phantom{5} e_b^u \sim B\pth{\dxaprime{u};\polx},\\
    &e_a^v \sim B\pth{\daprime{v};\plox}, \phantom{5} e_b^v \sim B\pth{\dxaprime{v};\polx}
\end{align*}
where \(\dxaprime{u} = n-2-\daprime{u}\) and \(\dxaprime{v} = n-2-\daprime{v}\).
The probability generating function of a binomially distributed random variable \(X\sim\textrm{Bin}(n,p)\) is given by \([1+p(z-1)]^n\). Thus we get the probability generating function of \(\beta\) as
\begin{align*}
    F_{\beta}(z) =& z^{\alpha}\pth{1+\plox\pth{z^{-1}-1}}^{\daprime{u}}\pth{1+\plox(z-1)}^{\daprime{v}}\\
    &\times\pth{1+\polx(z-1)}^{\dxaprime{u}} \pth{1+\polx\pth{z^{-1}-1}}^{\dxaprime{v}}
\end{align*}

\textbf{Probability generating function of the relative signature distance gamma}
\label{appendix:gamma}

Consider the random bipartite graphs \(B_a = (V,H;\Ea)\), \(B_b = (V,H;\Eb)\) distributed according to \((B_a,B_b) \sim ER(h,n;\pvec)\). For a given pair of vertices \(u,v \in V\) let us define the relative signature distance of \(u\) to \(v\) observed from \(\Gb\) as \(\gamma(u,v) = \norm{\SBa{v}-\SBb{u}} - \norm{\SBa{u}-\SBb{u}}\). \(\gamma(u,v)\) can be expressed as the sum of the contributions of each high-degree vertex \(w\in H\). The neighborhoods \(\Na{v},\Na{u}\) and \(\Nb{u}\) partition the set of high-degree vertices in 8 disjoint sets as given in \hyperref[fig:venn]{Fig. \ref*{fig:venn}}. We then have \(\gamma(u,v) = \sum_{w\in H} \ind{w\in H_3\cup H_4} - \ind{w\in H_1 \cup H_6}\).

\begin{figure}
\centering
\includegraphics[width=0.6\linewidth]{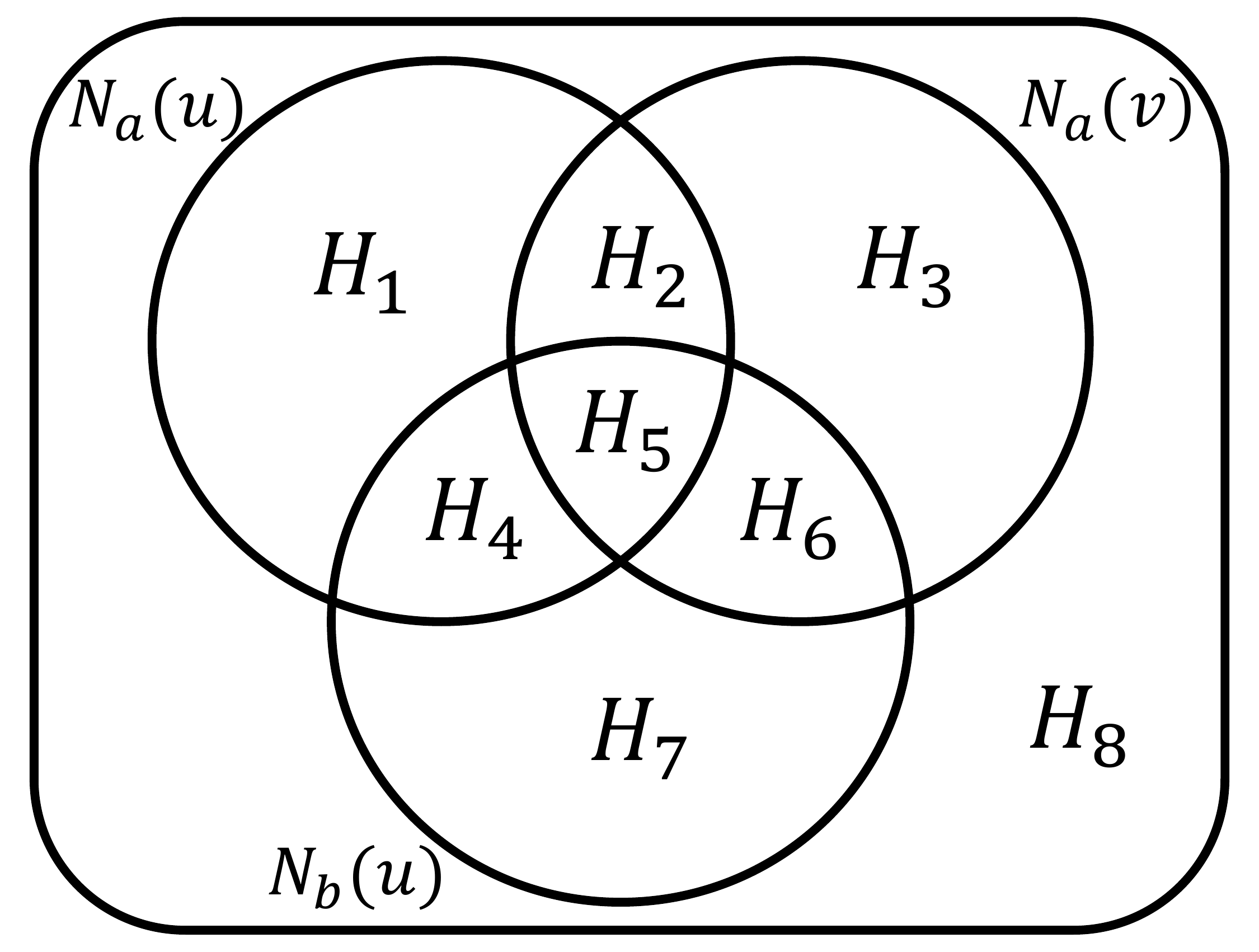}
\caption{Venn diagram representation of \(\Na{u}\), \(\Na{v}\) and \(\Nb{u}\)}
\label{fig:venn}
\end{figure}

Notice that for any \(w\in H\), \(\prob{w\in H_3\cup H_4} = \poo\plx + \pll \pox\) and \(\prob{w\in H_1 \cup H_6} = \plo\pox + \pol\plx\). In fact the random variables \(\Bacc{\ind{w\in H_3\cup H_4} - \ind{w\in H_1 \cup H_6}}_{w\in H}\) are mutually independent and identically distributed. 
Let us define \(\qo = \poo\plx + \pll \pox\) and \(\ql = \plo\pox + \pol\plx\)
This gives us the following generating function
\begin{align}
    \label{GF}
    F_{\gamma}(z) = \E{z^{\gamma(u,v)}} = \croc{1+\qo\pth{z-1}+\ql\pth{z^{-1}-1}}^h.
\end{align}

\bibliographystyle{ieeetr}
\bibliography{bibliography}

\end{document}